%% file: main.tex
\documentclass[letterpaper,twocolumn,10pt]{article}
\usepackage{usenix2019_v3}
\usepackage{caption}
\usepackage{subcaption}
\usepackage{float}
\usepackage{listings}
\usepackage{graphicx}
\usepackage{longtable}
\usepackage{wrapfig}
\usepackage{rotating}
\usepackage[normalem]{ulem}
\usepackage{capt-of}
\usepackage{algorithm}
\usepackage{algorithmic}
\usepackage{textcomp}
\usepackage{xcolor}
\usepackage{booktabs}
\usepackage{caption}
\usepackage{subcaption}
\usepackage{listings}
\usepackage{amsthm}
\usepackage{amsmath}
\usepackage{amssymb}
\usepackage{amsfonts}
\usepackage{pifont}
\newcommand{\xmark}{\textcolor{red}{\ding{55}}}%
\newcommand{\qmark}{\textcolor{blue}{\fontfamily{cyklop}\selectfont \textit{?}}}

\usepackage{hyperref}
\usepackage{url}
\usepackage{textcomp}
\usepackage{siunitx}
\include{math_commands.tex}

\title{Towards a Re-evaluation of Data Forging Attacks in Practice}

\begin{document}



\author{
{\rm Mohamed Suliman\textsuperscript{1,2}, Anisa Halimi\textsuperscript{1}, Swanand Ravindra Kadhe\textsuperscript{1}, Nathalie Baracaldo\textsuperscript{1}, Douglas Leith\textsuperscript{2}}\\
\textsuperscript{1} IBM Research \ \textsuperscript{2} Trinity College Dublin, The University of Dublin
} 

\maketitle

\begin{abstract}
Data forging attacks provide counterfactual proof that a model was
trained on a given dataset, when in fact, it was trained on
another. These attacks work by forging (replacing) mini-batches with
ones containing distinct training examples that produce nearly
identical gradients. Data forging appears to break any potential
avenues for data governance, as adversarial model owners may forge
their training set from a dataset that is not compliant to one that
is. Given these serious implications on data auditing and
compliance, we critically analyse data forging from both a practical
and theoretical point of view, finding that a key practical
limitation of current attack methods makes them easily detectable by
a verifier; namely that they cannot produce sufficiently identical
gradients. Theoretically, we analyse the question of whether two
distinct mini-batches can produce the same gradient. Generally, we
find that while there may exist an infinite number of distinct
mini-batches with real-valued training examples and labels that
produce the same gradient, finding those that are within the allowed
domain e.g. pixel values between 0-255 and one hot labels is a non
trivial task. Our results call for the reevaluation of the strength
of existing attacks, and for additional research into successful
data forging, given the serious consequences it may have on machine
learning and privacy.
\end{abstract}

\section{Introduction}
Data governance is becoming an increasingly important subject with the
rise of indiscriminate scraping of web data to train machine learning
models. Legislation such as the GDPR Framework
\cite{voigt2017eu,gdpr2016} and the CCPA \cite{ccpa2018} have been
introduced to prevent the use of sensitive personal information during
model training and reduce potential harm to the data owners that may
result from the proliferation of these models. Multiple stakeholders
exist: data owners want assurances that their copyrighted or private
information has not been used to train models without their consent,
while the model owners want to prove that their training data is
legally compliant.

Data forging
\cite{thudi2022necessity,kong2023can,baluta2023unforgeability,zhangverification}
has recently emerged as a new attack against machine learning models,
particularly against data governance. Data forging attacks appear to
provide counterfactual ``proof'' that a model was trained using a
particular dataset; when in reality, it was trained on another. In
effect, these attacks \emph{forge} one dataset into another
\emph{distinct} dataset that produces the same model. Given the true
training set \(D\), a forgery \(D'\) may be produced along with
sufficient proof claiming that \(D'\) is the real training set, not
\(D\). Thus, data forging attacks appear to throw into jeopardy the
question of ever determining exactly what data a model was trained on.

\begin{figure}
  \centering
  \includegraphics[scale=0.34]{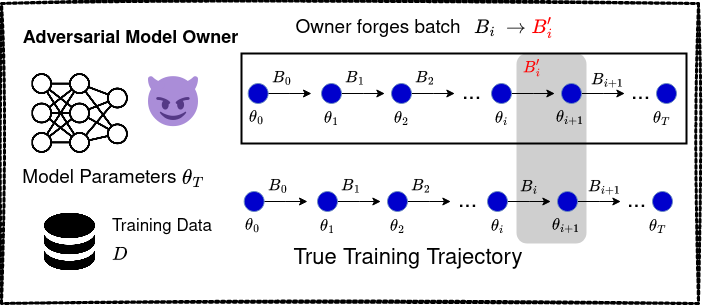}
  \caption{The adversarial model owner has access to the true
    execution trace denoting the sequence of model parameters and
    associated mini-batches. When performing a data forging attack,
    the model owner \emph{forges} (\emph{i.e.}, replaces) a batch $B_i$ such
    that $(\tbf{x}, \tbf{y}) \in B_i$ with a different batch $B'_i$ not
    containing $(\tbf{x}, \tbf{y})$ such that the resulting models are
    nearly identical.}
  \label{fig:data_forging}
  \vspace{-15pt}
\end{figure}
 
\noindent\textbf{Data Forging.} At a high level, given a machine learning model's
true training dataset $D$, a data forging attack produces a claim that
the model is actually the result of training on a \emph{different}
dataset $D'$. To achieve this, data forging attacks rely on the fact
that supervised training uses iterative algorithms such as Stochastic
Gradient Descent (SGD). These algorithms produce a sequence consisting
of model parameters (or \emph{checkpoints}) and associated
mini-batches, starting with the random initialisation and ending with
the final trained model parameters.  Such a sequence, or
\emph{execution trace}, may be \emph{verified} by a third party, who
reproduces each checkpoint in the sequence, using the mini-batch and
parameters from the previous step, and ensures that what they have
recomputed and the current checkpoint are nearly identical, i.e., the
$\ell_2$ distance between the two sets of model parameters is below a
given \textit{small} non-zero error threshold \(\epsilon\). A non zero
threshold may be allowed by the verifier due to noise that stems from
hardware and software factors. Specifically, there are small numerical
deviations between repeated recomputations of any particular
checkpoint due to benign noise (see Section
\ref{sec:reproduction-errors} for a detailed discussion on the source
of these errors). Importantly, if an execution trace passes
verification i.e., the recomputation of every step is below the
verifier's chosen threshold, then that trace is valid ``proof'' for
the verifier that the data used to train the model is what occurs in
the mini-batches present in the trace and no other (provided the
initial model \(\theta_0\) is random). This concept is formally known as
Proof of Learning (PoL), introduced by \cite{jia2021proof}. An
adversary performing a data forging attack replaces one or more
mini-batches in the execution trace with \textit{forged} (\emph{i.e.},
different) mini-batches that produce nearly identical gradient
updates, and falsely claims that \emph{(i)} the model checkpoints
present in the execution trace are actually the result of training on
these forged mini-batches, not the original mini-batches that were
replaced, and \emph{(ii)} the observed recomputation error by the
verifier is simply a result of hardware and software factors.

\noindent\textbf{Privacy Implications.}
Existing works \cite{thudi2022necessity,kong2023can,zhangverification}
have noted that data forging has clear implications for two large
concepts within machine learning and privacy; namely \emph{(i)}
membership inference, and \emph{(ii)} machine unlearning. Membership
inference attacks (MIAs)
\cite{shokri2017membership,carlini2022membership,hu2022survey} have
been developed in order to determine whether a particular data record
was part of the model's training dataset. Suppose that the training
example \((\tbf{x},\tbf{y})\) is indeed a member of the model's
training set; an adversary may refute this claim by providing a forged
execution trace such that \((\tbf{x},\tbf{y}) \notin B_i\), for all
mini-batches \(B_i\) used during the execution trace (see
Fig.~\ref{fig:data_forging}). If such an execution trace passes
verification, then the adversary has successfully refuted the
membership inference claim as they have provided valid proof for the
verifier that \((\tbf{x}, \tbf{y})\) was not used to train their
model. This is because in none of the steps of the execution trace
does \((\tbf{x}, \tbf{y})\) occur as a member of the current step's
mini-batch.

Performing such an attack is also equivalent to claiming that
\((\tbf{x},\tbf{y})\) has been exactly ``unlearnt''. The concept of
\textit{exact unlearning} \cite{bourtoule2021machine}, refers to the
idea of unlearning i.e., removing the influence of, a particular
example \((\tbf{x}, \tbf{y})\) from the model, which involves
retraining from scratch using the dataset
\(D' := D \backslash \{(\tbf{x},\tbf{y})\}\). The forged execution trace that
passes verification provides valid ``proof'' to the verifier that the
final model parameters are the result of training on a dataset that
does not contain \((\tbf{x}, \tbf{y})\); the definition of exact
unlearning. Importantly, the adversary does not have to perform any
actual re-training, resulting in two equally valid claims that appear
to contradict each other; the model has been both trained and not
trained using \((\tbf{x},\tbf{y})\). Machine unlearning
\cite{xu2023machine,bourtoule2021machine,cao2015towards} has emerged
as a response to GDPR's \cite{voigt2017eu} \emph{right to be
  forgotten} framework, and prior work
\cite{thudi2022necessity,zhangverification} note how data forging can
undermine the concept, and in its wake, call for the research and
development of auditable machine unlearning. We argue that the
seriousness of these implications hinges on the performance of data
forging attacks, i.e., \emph{how identical are the model updates
  between original and forged mini-batches produced by these attacks?}
This metric, known as approximation error, is given by the \(\ell_2\)
norm between the forged model update and the original present in the
execution trace.

Within the data forging literature, there exists a dichotomy of
assumptions on how small this approximation error should be in order
to fully realise the implications discussed previously. Prior data
forging works \cite{kong2023can,thudi2022necessity,zhangverification}
argue that their attacks produce approximation errors that are
indistinguishable from benign reproduction errors. On the other hand,
Baluta et al. \cite{baluta2023unforgeability} argue against the
acceptance of any error, citing that floating point operations are
deterministic and that there should be zero error when recomputing
execution traces. Baluta et al. go on to prove that zero error, or
\emph{exact} data forging, is not possible using existing attack
methods. As a result, they find that existing data forging attacks do
not pose any risk to membership inference or machine unlearning as
they are easily detectable by the presence of non zero error. While
the zero error threshold scenario posed by Baluta et
al. \cite{baluta2023unforgeability} is possible if the original
hardware and software are available during verification (and may be
necessary in certain high stakes scenarios such as medical data), it
is unrealistic to always impose such a strict setting in
practice. Particularly when the original hardware and software is not
available, it is not practical, and often impossible, to avoid benign
hardware reproduction errors. We argue that data forging attacks fully
realise their implications only if their approximation errors are of
the same order of magnitude as these benign reproduction errors. Prior
data forging works
\cite{kong2023can,thudi2022necessity,zhangverification} do not justify
the level of approximation error generated by their attacks, and
largely ignore this pertinent question regarding their performance.

\noindent\tbf{Contributions.} We list our main contributions to data
forging as the following:

\begin{enumerate}
\item We provide the first comprehensive analysis of the magnitude of
  reproduction errors for fully connected, convolutional, and
  transformer neural networks architectures trained on tabular, image,
  and textual data.

\item We find that the approximation errors produced by existing data
  forging attacks are too large, often several orders of magnitude
  larger than our observed reproduction errors. This makes existing
  attacks easily detectable by a verifier, nullifying their privacy
  impact in practice.
  
\item The current theoretical analysis of data forging in Baluta et
  al. \cite{baluta2023unforgeability} is restricted to the existing
  instantiations of data forging attacks and necessitates the
  investigation of broader potential attack strategies. To tackle
  this, we extend their analysis, addressing more general questions
  regarding the existence of distinct mini-batches that produce the
  same gradient and how they may be constructed. 
  Focusing on fully connected neural networks, we formulate the
  problem of exact data forging as a system of linear equations, where
  each solution corresponds to a forged mini-batch. Of the infinite
  solutions that may exist, finding those that correspond to a
  distinct, valid mini-batch i.e., one that falls within the domain of
  allowed training examples (e.g., in the case of images, pixel values
  between \(0-255\) and one hot labels) is a non-trivial task, and one
  that we conjecture may be computationally infeasible, even for
  relatively simple linear models such as logistic regression.
\end{enumerate}

As a result, we call for a re-evaluation of existing attacks, and for
further research into this nascent attack vector against machine
learning.

\section{Background}
\label{sec:backkground}

\subsection{Execution Traces}
\label{sec:execution-traces}

Supervised machine learning is a process to learn a \textit{model}, in
particular, a parameterized function
$f_{\theta}: \mathcal{X} \to \mathcal{Y}$. The parameters are typically optimized by applying
iterative methods such as Stochastic Gradient Descent (SGD) to a
training set. At the \(i\)-th gradient descent step, the next set of
model parameters \(\theta_{i+1}\) is calculated from the previous set of
parameters \(\theta_i\) and a mini-batch \(B_i\) consisting of \(b\)
training examples sampled from the dataset
\(D = \{(\tbf{x}^{(i)}, \tbf{y}^{(i)})\}_{i=1}^N\). For mini-batch
SGD, we have that
\(\theta_{i+1} = g(\theta_i, B_i) = \theta_i - \eta\nabbi\). Starting with a random
initialisation \(\theta_0\), the final model \(\theta_T\) is the result of
performing \(T\) total gradient descent steps iteratively. This
results in an \emph{execution trace} (or simply, a trace) consisting
of a sequence of tuples \(\{(\theta_i,B_i)\}_{i=0}^{T}\), that capture the
\emph{training trajectory} of the final model \(\theta_T\) (note that
\(B_T = \emptyset\)).

\subsection{Reproduction Errors}
\label{sec:reproduction-errors}
Recent work by Schl\"{o}gl et al. \cite{schlogl2023causes} found that
runtime optimisations of machine learning frameworks can result in
non-deterministic numerical deviations in the outputs of neural
networks. Schl\"{o}gl et al. cite auto-tuning \cite{grauer2012auto} as
the root cause. This is a process whereby microbenchmarks executed
just before inference determine what optimisations to apply to the
given operation e.g., the results of the microbenchmarks determine
which GPU kernel gives the best performance on a given hardware and
input shape. A well known property of floating point arithmetic is
that addition is not associative, that is, it is not necessarily the
case that \((a+b) + c = a + (b+c)\), but may be off by a small
error. These auto-tuning optimisations can lead to changes in the
aggregation order of intermediate results from kernel to kernel,
resulting in deviations in the final output of a neural network, even
for the same input. In the context of execution traces for machine
learning models, this means that given the same random initialisation
\(\theta_0\) and sequence of mini-batches
\(B_0,B_1,B_2\dots,B_{T-1}\), the resulting parameters
\(\theta_1,\theta_2,\theta_3,\dots,\theta_{T}\) will differ between repeated
recomputations due to auto-tuning.

\subsection{Data Forging}
\label{sec:data-forging}
Thudi et al.'s \cite{thudi2022necessity} concepts of
\emph{forgeability} and data forging attempt to produce distinct
execution traces for a given model \(\theta_T\) using the same checkpoints,
but with different e.g., \emph{forged} mini-batches. At a high level,
any mini-batch from an execution traces may be \emph{forged}
(\emph{i.e.}, replaced) with another containing different training
data, yet still produce a nearly identical gradient descent
update. \emph{Data forging attacks} against execution traces attempt
to construct, for a given mini-batch \(B_i\), a forged (\emph{i.e.},
distinct) counterpart \(B'_i\) such that the next model in the
execution trace, when calculated using \(B'_i\) is \emph{nearly}
identical to \(\theta_{i+1}\) \emph{i.e.}, the approximation error
\(\|\theta_{i+1} - g(\theta_i,B'_i)\|_2 \) is minimal. The attacks proposed in
\cite{thudi2022necessity,kong2023can} construct forged mini-batches by
randomly sampling \emph{other} training examples from the
dataset. These attacks were introduced in the context of
\emph{deleting} a set \( U \subset D\) of training examples from an
execution trace, and work by \emph{greedily searching} over the set of
training examples $D \backslash U$. A \textit{greedy search} attack works in
three steps \emph{(i)} sample $n$ data points uniformly from
$D \backslash U$, \emph{(ii)} sample $M$ mini-batches
$\{\hat{B}_1,\dots,\hat{B}_M\}$ uniformly from the selected $n$ data
points, and \emph{(iii)} out of the $M$ mini-batches
$\{\hat{B}_1,\dots,\hat{B}_M\}$, select the mini-batch $\hat{B}_j$
that minimises \(\|\theta_{i+1} - g(\theta_i,\hat{B}_j)\|_2\), and output the
forged mini-batch $B'_i = \hat{B}_j$.

Thudi et al. introduced this type of forging approach, and Kong et
al. \cite{kong2023can} improve upon its efficiency, with similar
performance. Recently, Zhang et al. \cite{zhangverification} augmented
this approach by choosing to instead replace data points
\((\tbf{x}, \tbf{y}) \in U\) with their class-wise nearest neighbour
from \(D \backslash U\), with also similar performance.

\section{Data Forging Threat Model}
\label{sec:data-forging-threat}

At its core, the danger of data forging, and why it must be taken
seriously, ultimately stems from the uncertainty it casts over exactly
what data was used to train a model. This can allow an adversary
(\emph{i.e.}, a malicious model owner) to claim their model was
trained on one dataset, when in fact, it was trained on
another. Motivating applications include models that were trained on
copyrighted and/or sensitive information.

In Algorithm \ref{alg:1}, we formulate the threat of data forging as
game between an adversary \(\mathcal{A}\) and a verifier
\(\mathcal{V}\).  Our formulation captures the setup proposed in previous data
forging works
\cite{kong2023can,thudi2022necessity,baluta2023unforgeability,zhangverification}. We
allow for a non zero error between recomputed and original checkpoints
in contrast to the data forging game defined by
\cite{baluta2023unforgeability}, who required \(\epsilon = 0\), a threat
model which is much stricter than ours, and one that may not be
realistic in all scenarios due to the existence of floating point
errors (see Section \ref{sec:reproduction-errors} for a more detailed
discussion).

\begin{algorithm}
  \begin{algorithmic}[1]
    \REQUIRE Adversary \(\mathcal{A}\), Verifier \(\mathcal{V}\), Training algorithm \(\mathcal{T}\), Dataset \(D\), error threshold \(\epsilon\), batch size \(b\)
    \ENSURE ACCEPT or REJECT
    \STATE \(\mathcal{V}\) produces execution trace \(\{(\theta_i,B_i)\}_{i=0}^T \gets \mathcal{T}(D,b)\).
    \STATE \(\mathcal{V}\) chooses a random index \(t \in \{0,1,2,3,...,T-1\}\). 
    \STATE \(\mathcal{A}\) forges mini-batch \(B_t\) into \(B'_t \gets \mathcal{A}(\{(\theta_i,B_i)\}_{i=0}^T,D,t)\).
    \IF {\((B'_t = B_t)\) or \((\exists (\tbf{x}', \tbf{y}') \in B'_t \mbox{ s.t. } (\tbf{x}', \tbf{y}') \notin \mathcal{X} \times \mathcal{Y}) \)}
    \STATE \tbf{return} REJECT
    \ENDIF
    \STATE \(\mathcal{V}\) calculates \(\theta^{\mathcal{V}}_{t+1} \gets g(\theta_t, B'_t)\) on their own hardware.
    \STATE \(\mathcal{V}\) calculates the error \(\epsilon_{t+1} \gets \|\theta_{t+1} - \theta^{\mathcal{V}}_{t+1}\|_2\)
    \IF {\((\epsilon_{t+1} > \epsilon) \)} 
    \STATE \tbf{return} REJECT
    \ELSE
    \STATE \tbf{return} ACCEPT
    \ENDIF
\end{algorithmic}
\caption{\label{alg:1} The data forgery game between \(\mathcal{A}\) and \(\mathcal{V}\).}
\end{algorithm}

The game plays as follows: the verifier begins by first running the
training algorithm \(\mathcal{T}\) e.g., stochastic gradient descent on dataset
\(D\) for \(T\) steps, and produces execution trace
\(\{(\theta_i,B_i)\}_{i=0}^T\), saving every checkpoint. The verifier then
challenges the adversary to forge a random batch \(B_t\) from the
execution trace. In concrete data forging attack settings, batch
\(B_t\) may contain legally problematic or sensitive information; data
that the adversary must forge with approximation error less than the
verifier's chosen threshold \(\epsilon\), in order to successfully claim the
data in that batch was not used \emph{i.e.,} falsely \emph{delete}
from their model's training data. The adversary has access to the full
trace and the dataset, and is capable of instantiating any data
forging attack algorithm in order to produce the corresponding forged
mini-batch \(B'_t\), which must satisfy the following initial
conditions:

\begin{itemize}
\item \(B'_t \ne B_t\). Adversaries participating in the game must
  produce a forged mini-batch that is \emph{different} from the
  original in at least one example.
\item
  \((\tbf{x}', \tbf{y}') \in \mathcal{X} \times \mathcal{Y}, \ \ \forall (\tbf{x}', \tbf{y}') \in
  B'_t\). This condition constrains each forged training example
  \((\tbf{x}', \tbf{y}')\) to be within the domain of the training
  context \emph{e.g.} pixel values between \(0-255\) and one hot
  labels. Otherwise, \(B'_t\) will be rejected and the adversary loses
  the game.
\end{itemize}

Finally, The verifier computes
\(\theta_{t+1}^{\mathcal{V}} := g(\theta_t, B'_t)\) using the forged mini-batch
\(B'_t\) produced by the adversary and calculates the \(\ell_2\) error
\(\epsilon_{t+1} := \|\theta_{t+1} - \theta_{t+1}^{\mathcal{V}}\|_2\) between their recomputed
checkpoint, and the one present in the trace. The game ends with a
REJECT result if \(\epsilon_{t+1} > \epsilon\), otherwise the game ends with ACCEPT
\emph{i.e.}, the adversary wins. The reproduction error threshold
\(\epsilon\) is a parameter of the game that defines how large the
\(\ell_2\) between a recomputed model and the original model may be. We
introduce the notion of \emph{\tbf{stealthy data forging}} which
refers to the production of a forged mini-batch \(B'_t\) that wins the
above data forging game. Concretely, this means that attacks that
produce a \(B'_t\) whose approximation error is less than or equal to
the chosen threshold \(\epsilon\) are deemed \emph{stealthy}.

\section{Are Current Data Forging Attacks Stealthy?}
\label{sec:are-data-forging}
Previous data forging works
\cite{kong2023can,thudi2022necessity,zhangverification} have employed
data forging as a means of adversarially \emph{deleting} data from a
model's execution trace \emph{i.e.}, for a given set of training
examples \(U \subset D\), these attacks forge the true dataset \(D\) into
another \(D'\) such that \(D' \cap U = \emptyset\). Kong et
al.~\cite{kong2023can} note that if an adversary can delete training examples
from an execution trace by replacing them with forged examples, they
can provably refute a membership inference claim against
them. Additionally, \cite{thudi2022necessity,zhangverification}
observe that if an adversary performs the same replacement, they can also claim
to have ``unlearnt'' those training examples. These implications of
data forging are only fully realised if the attacks are
\emph{stealthy} i.e., the adversary wins the game with an approximation error
lower than the error threshold chosen by the verifier.

Theoretically, Thudi et al. \cite{thudi2022necessity} have proven the
existence of an adversary that can win the forging game defined in the
previous section for any \(\epsilon > 0\), given that they have access to the
underlying data distribution \(\mathcal{D}\), and have freedom over the choice
of batch size \(b\). In particular, they prove that in the limit
\(b \rightarrow \infty \), the probability that mini-batches from any datasets
\(D\) and \(D'\) sampled i.i.d from \(\mathcal{D}\) can be forged approaches 1
as \(b \rightarrow \infty\) (see Theorems 2 and 3 from \cite{thudi2022necessity}
regarding probabilistic forging). Of course, in practice, the
adversary cannot sample indefinitely\footnote{Additionally, there only
  exists a finite number of images consisting of discrete pixel
  values.}, from, nor do they have access to \(\mathcal{D}\), however, Thudi et
al. demonstrate the practicality of forging by developing a data
forging attack algorithm that can win the forging game for
\(\epsilon \ll 1\) (see Sec. \ref{sec:data-forging} for a description of their
algorithm, as well as others in the literature). However, they do not
consider the pertinent question regarding whether the approximation
errors produced by their attack are comparable to reproduction errors
induced by floating point noise.

In practice, an informed verifier will choose their error threshold
\(\epsilon\) to be no larger than observed reproduction errors. As a result,
for a data forging attack to be stealthy, its approximation errors
must be on the same order of magnitude as reproduction
errors. Stealthiness, we argue, derives from the fact that the
approximation errors would be indistinguishable from the regular noise
that pervades floating point computation. In order to answer the
question of whether existing attacks are stealthy, we first conduct
experiments to determine the magnitude of reproduction errors, and
compare this to the approximation errors produced by existing
attacks. In the next section, we perform repeated recomputations of
execution traces, and measure the \(\ell_2\) distance between repeated
recomputations, in order to characterise the magnitude of reproduction
errors. Subsequently, we compare the magnitude of our observed
reproduction errors to the magnitude of the approximation errors
produced by the data forging attacks developed in
\cite{kong2023can,thudi2022necessity,zhangverification}. \tbf{Our
  findings suggest that existing data forging attacks are not
  stealthy} i.e., reproduction errors are often orders of magnitude
smaller than approximation errors produced by existing data forging
attacks.

\subsection{What is the Magnitude of Benign Reproduction Errors?}
\label{sec:what-magn-benign}

In order to quantify the magnitude of reproduction errors, we perform
experiments where a model is trained on a dataset for \(T\) total
steps, recording its execution trace. We then recompute this trace, on
the same and also different hardware, and measure the distance between
every recomputed checkpoint and the corresponding original
checkpoint. We repeat the verification process \(10\) times. The
magnitude of the error between repeated recomputations of the same
checkpoint is called the \emph{reproduction error} and is calculated
using the \(\ell_2\) norm
\(\epsilon_{repr} := \arg\max_i \epsilon_{i+1}\) \emph{i.e.}, the reproduction error
is the largest observed \(\ell_2\) distance between recomputations. In
Table \ref{tab:models}, we summarise the models and datasets used in
our experiments, which includes a fully connected neural networks
(FCN) with 1 ReLU hidden layer consisting of 100 neurons, a small and
large convolutional neural network, and a decoder-only
transformer. These models are trained on the Adult, MNIST, CIFAR10,
and IMDB datasets respectively, covering classification, object
detection, and sentiment analyis tasks across tabular, image, and
textual data modalities. We keep the learning rate fixed at
\(\eta = 0.01\). For all our experiments, we use TensorFlow v2.15.0, CUDA
v12.2, and CUDNN v8.9. In addition to \texttt{float32} training, we
also investigate the effect of \texttt{float16} and \texttt{bfloat16}
mixed-precision training on reproduction errors.

\begin{table}[h]
  \centering
  \caption{Models and Datsets.}
  \label{tab:models}
  \begin{tabular}{||l | l ||}
    \hline

    Model / Dataset & \# Total Params\\
    \hline
    FCN / Adult &  \(1,601\) \\
    LeNet / MNIST & \(61,706\) \\
    Transformer / IMDB& \(657,758\) \\
    VGGmini / CIFAR10 & \(5,742,986\) \\
    \hline
  \end{tabular}
\end{table}

For a given execution trace (e.g., a trace produced for a given model,
dataset, batch size \(b\), and on a given hardware platform), our
experiments proceed as follows:

\begin{enumerate}
\item First, we choose the hardware platform where the execution trace
  will be recomputed \emph{i.e.}, this is either the same GPU it was
  generated on or a different one.

\item We recompute each checkpoint by loading \(\theta_{i}\) from 
  the execution trace and training on the corresponding mini-batch \(B_i\), 
  also taken from the execution trace.
  
\item We measure the distance between our recomputed checkpoint and
  the corresponding checkpoint \(\theta_{i+1}\) present in the execution
  trace.
\item We repeat this process \(10\) times and report the largest
  observed distance for each recomputed checkpoint, as well as
  variance.
\end{enumerate}

The question of what might affect the magnitude of reproduction errors
boils down to the question of what might affect the results of the
auto-tuning microbenchmarks that ultimately decide which floating
point operations are done in which order. It is possible (however
unlikely), that two training runs, starting with the same random
initialisation and using the same sequence of mini-batches, may
produce the exact same sequence of checkpoints (\emph{i.e.}, zero
reproduction error) simply by chance. In reality, reproduction errors
occur because the likelihood that the floating point operations are
done in the same order after auto-tuning between two different
training runs is very small. What goes into the decision of gpu kernel
can depend on a host of reasons such as the model architecture,
dataset, the particular hardware and software, the other processes
that are running on the machine at the time, etc. A comprehensive
analysis of all these factors that affect reproduction error is beyond
the scope of this work, however we pick the same models and datasets
as previous works e.g. LeNet/MNIST and VGGmini/CIFAR10 as used in
\cite{baluta2023unforgeability,thudi2022necessity,kong2023can} (and
additionally FCNs and decoder-only transformers), in order to
determine a value for the magnitude of reproduction errors such that
we may compare their attack's approximation errors to. Previous work
by \cite{pham2020problems} considered the effect non determinism
introduced by GPU training on the final accuracy of the trained models
(see Appendix \ref{sec:tbfacc-meas} our accuracy measurements). Our
experiments regarding reproduction errors attempt to determine the
magnitude of these errors as we vary model size and type, data
modality, floating point precision, batch size, and GPU model.

\subsubsection{Batch size, model size and data modality do not appear
  to greatly affect reproduction error}
\label{sec:batch-size-model}

To characterise the effect of batch size on reproduction errors, we
train our models with batch sizes
\(b = \{50, 100, 500, 1000, 2000\}\), covering 3 orders of
magnitude. Figure \ref{fig:bs} provides the largest observed
reproduction errors during the recomputation of execution traces for
each of the 4 considered model setups. We see that the reproduction
error remains constant for all 4 setups, varying by only one order of
magnitude at most. The largest observed reproduction error was for the
transformer model, which was no larger than \(\num{1e-6}\) in
magnitude for \(b = 50\). Additionally, the models considered span 4
orders of magnitude in terms of the number of parameters (see Table
\ref{tab:models}), and we find no correlation between model size and
reproduction error; the transformer model consisting of 600K
parameters produces larger reproduction errors than VGGmini which has
5.75M parameters. From our experiments the effect of data modality is
unclear i.e., images and tabular data produce similar levels of
reproduction errors whereas textual data produces slightly larger 
errors (larger by a single order of magnitude).

\begin{figure}
  \centering
  \includegraphics[scale=0.5]{ 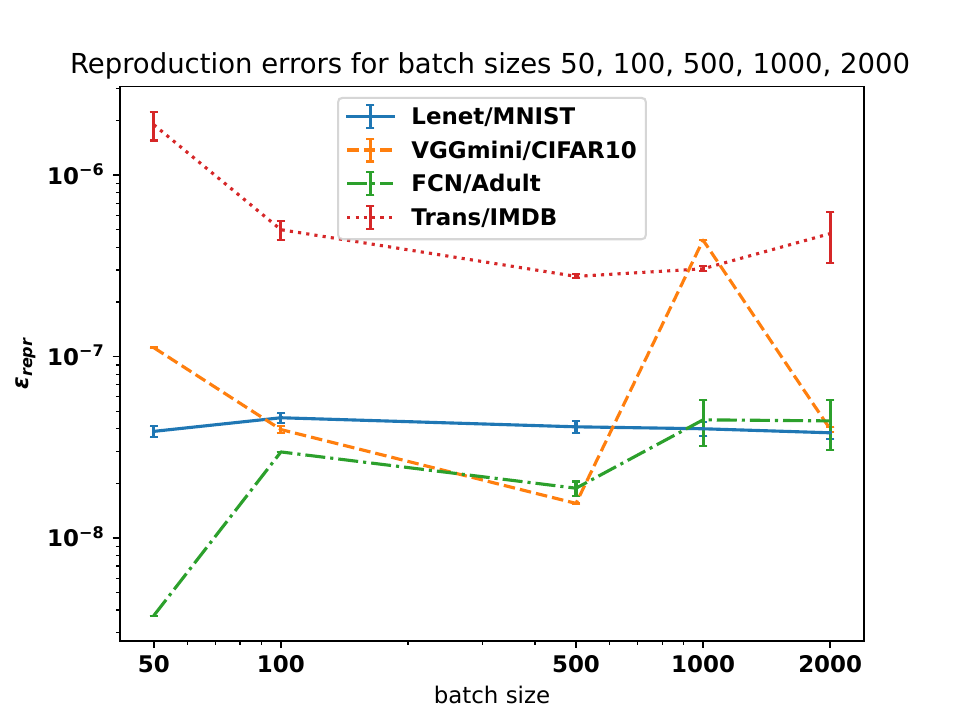 }
  \caption{The observed largest reproduction error across multiple
    recomputations across batch size and model type and size.}
  \label{fig:bs}
\end{figure}

\subsubsection{Mixed precision training can result in larger
  reproduction errors}

Mixed precision training involves using lower precision floating point
types such as \texttt{float16} and \texttt{bfloat16}\footnote{Both
  formats take up 2 bytes, however \texttt{float16} uses 10 bits for
  the mantissa, whereas \texttt{bfloat16} uses 7.} during training in
order to increase performance on modern GPUs. Given their lower
precision, it is possible that reproduction errors for lower precision
floating point numbers can be larger than those for the default
\texttt{float32}, as investigated previously in Section
\ref{sec:batch-size-model}.

To evaluate this claim, we re-train our models using both
\texttt{float16} and \texttt{bfloat16} specifications and provide the
largest observed reproduction errors in Table
\ref{tab:floating-point}. Generally, lower precision training appears
to result in larger reproduction errors as compared to when training
with the more precise \texttt{float32}. We find that these errors may
be up to 3 orders of magnitude larger e.g. LeNet/MNIST for batch size
\(b=100\) results in reproduction errors \(\sim \num{1e-8}\) with
\texttt{float32} training, however can be as large as
\(\sim \num{1e-5}\) with mixed precision \texttt{bfloat16} training.

\begin{table}
  \begin{subtable}{\columnwidth}
    \centering
    \begin{tabular}{||c | c ||}
      \hline
      Model/Dataset & Largest \(\epsilon_{repr}\) \\
      \hline 
      LeNet/MNIST   & \(\num{8.96e-6} (\pm \num{8.64e-7}) \) \\
      VGGmini/CIFAR10 & \(\num{7.35e-5} (\pm \num{1.11e-6})\) \\
      FCN/Adult & \(\num{4.2e-6} (\pm \num{1.94e-7})\) \\
      Trans/IMDB & \(\num{4.74e-5} (\pm \num{2.83e-5})\) \\
      \hline
    \end{tabular}
    \caption{float16}
    \label{tab:fp16}
  \end{subtable}
  \hfill
  \begin{subtable}[h]{\columnwidth}
    \centering
    \begin{tabular}{||c | c||}
      \hline
      Model/Dataset & Largest \(\epsilon_{repr}\) \\
      \hline 
      LeNet/MNIST   & \(\num{6.6e-5} (\pm \num{4.77e-6}) \) \\
      VGGmini/CIFAR10 & \(\num{5.97e-4} (\pm \num{2.54e-5})\) \\
      FCN/Adult & \(\num{1.22e-04} (\pm \num{5.14e-5})\) \\
      \hline
    \end{tabular}
    \caption{bfloat16}
    \label{tab:bf16}
  \end{subtable}
  \caption{The largest observed $\epsilon_{repr}$ when the training and
    recomputation are done using mixed-precision. }
  \label{tab:floating-point}
\end{table}


\subsubsection{Recomputation on different hardware can result in
  larger reproduction errors}
\label{sec:recomp-diff-hardw}
We now attempt to characterise the magnitude of benign hardware errors
when the recomputation is done on a different GPU than the one that
produced the execution trace. In Table \ref{tab:cross-hardware}, we
report the largest observed reproduction error for cross hardware
recomputation. We consider two different GPUs in our experiments: a
RTX 4090 and a Tesla V100. We find that recomputing execution traces
on different hardware can produce slightly larger levels of error
compared to when recomputation is done on the original hardware, In
our cross-hardware experimentation, we found that the error increases
only by a single order of magnitude. On both GPUs, auto-tuning is
turned on.

In most machine learning applications, auto-tuning is turned on as it
boosts performance, however, widely used machine learning frameworks
such as TensorFlow \cite{abadi2016tensorflow} and PyTorch
\cite{paszke2019pytorch} provide an option to turn it off\footnote{See
  \url{https://www.tensorflow.org/api_docs/python/tf/config/experimental/enable_op_determinism}
  and \url{https://pytorch.org/docs/stable/notes/randomness.html}}. We
found that turning auto-tuning off results in zero error when
recomputing execution traces using the same hardware and software, as
the floating point operations occur in a deterministic order. However,
different hardware platforms may not share the same implementations of
algorithms \emph{i.e.}, they may differ in approach and aggregation
order, resulting in unavoidable reproduction error, even if
auto-tuning is turned off. Table \ref{tab:ch-lenetb100_determ} shows
how when auto-tuning is off e.g., in deterministic training, there is
zero error when the execution trace is generated and recomputed on the
same hardware. However, when the generating and recomputing hardware
differ, we found that the reproduction error is non zero, and during
our experiments, was larger than when auto-tuning is turned on.


\begin{table}
  \begin{subtable}{\columnwidth}
    \centering
    \begin{tabular}{||c | c | c||}
      \hline
      & RTX 4090 & Tesla V100 \\
      \hline 
      RTX 4090   & \(10^{-8}\) & \(10^{-7}\)\\
      \hline
      Tesla V100 & \(10^{-7}\) & \(10^{-8}\)\\
      \hline
    \end{tabular}
    \caption{LeNet/MNIST \(b = 100\)}
    \label{tab:ch-lenetb100}
  \end{subtable}
  \hfill
  \begin{subtable}[h]{\columnwidth}
    \centering
    \begin{tabular}{||c | c | c||}
      \hline
      & RTX 4090 & Tesla V100 \\
      \hline 
      RTX 4090   & 0 & \(10^{-6}\)\\
      \hline
      Tesla V100 & \(10^{-6}\) & 0\\
      \hline
    \end{tabular}
    \caption{Deterministic LeNet/MNIST \(b = 100\)}
    \label{tab:ch-lenetb100_determ}
  \end{subtable}
  \hfill
  \begin{subtable}[h]{\columnwidth}
    \centering
    \begin{tabular}{||c | c | c||}
      \hline
      & RTX 4090 & Tesla V100 \\
      \hline 
      RTX 4090   & \(10^{-8}\) & \(10^{-7}\)\\
      \hline
      Tesla V100 & \(10^{-7}\) & \(10^{-8}\)\\
      \hline
    \end{tabular}
    \caption{LeNet/MNIST \(b = 1000\)}
    \label{tab:ch-lenetb1000}
  \end{subtable}    
  \caption{The largest observed $\epsilon_{repr}$ when recomputation is done
    on different hardware. For a given execution trace setup
    (\emph{i.e.}, each subtable), each row indicates the GPU the
    execution trace was generated on and the column indicates which
    GPU the trace was recomputed on.  In
    Tab.~\ref{tab:ch-lenetb100_determ}, the diagonal entries are zero
    because deterministic training results in zero error when
    recomputing on the same hardware. }
  \label{tab:cross-hardware}
\end{table}

\subsubsection{Discussion}
In this section, we have provided our results on observed reproduction
errors as we vary models (large and small), datasets, batch sizes,
floating point precision, and recomputation hardware, however we
refrain from making any statement about what directly results in the
magnitude of the reproduction error for a given setup. This is due to
the closed source nature of CUDA (NVIDIA's proprietary GPU driver
code), making any grounded determinations as to what is going on
``under the hood'' of the GPUs extremely difficult. 
However, our results and methodology represent a principled
approach for verifiers to calculate suitable error thresholds for the
verification of execution traces. Using these results, in the next
section we evaluate the \emph{stealthiness} of data forging attacks.

\subsection{Evaluating the Performance of Existing Data Forging
  Attacks: How Small are their Approximation Errors?}
\label{sec:when-data-forging}

In this section, we perform data forging attacks against the execution
traces we generated in the previous section, and report their
approximation error. Our results highlight two important factors that
affect the approximation error; namely \emph{(i)} the forging fraction
(\emph{i.e.}, the number of examples in a mini-batch that need to be
replaced), and \emph{(ii)} the stage of training at which forging is
taking place.

\subsubsection{\tbf{Smaller forging fractions result in smaller approximation errors}}
Recall that the goal of the adversary is to replace the original
mini-batch \(B_i\) with a forged mini-batch \(B'_i\) such that
\(U \cap B'_i = \emptyset\). The difficulty of producing the forged mini-batch
\(B'_i\) can be viewed as being proportional to the cardinality of
\(B_i \cap U\). If \(\#(B_i \cap U) = 1\), \emph{i.e.}, only one example
needs to be forged, then the adversary can keep all the other \(b-1\)
examples the same in \(B'_i\) and only forge the single training
example. For larger values of \(b\), the effect of a single example on
the average gradient \(\nabla_{\theta} \mathcal{L}(B'_i)\) becomes negligible; regardless
of what the training example is replaced with, its effect on the
average gradient is essentially ``drowned out'' by the other \(b-1\)
training examples that are shared between \(B_i\) and \(B'_i\). Given
a fixed number of training examples in \(B_i\) that needs to be
replaced, for larger and larger batch sizes (\emph{i.e.},
\emph{smaller forging fractions} \(\frac{\#(B_i \cap U)}{\#(B_i)}\)), we
expect the approximation error to get smaller. However, if
\(\#(B_i \cap U) = \#(B_i) = b\), then forging is much more difficult, as
every example needs to be replaced. Figure \ref{fig:forging-frac}
shows that this expected behaviour is indeed observed when forging
both LeNet/MNIST and VGGmini/CIFAR10 traces. We see that the attacks
perform best (\emph{i.e.}, produce the smallest approximation errors)
when the forging fraction
\(\frac{\#(B_i \cap U)}{\#(B_i)} = \frac{1}{b}\), and degrade quickly
with increased forging fractions. We also observe that Thudi et
al. \cite{thudi2022necessity} and Zhang et
al.\cite{zhangverification}'s attacks have similar performance (see
Appendix \ref{sec:tbfc-data-forg} for a comparison).

\begin{figure}[htb!]
  \centering
  \begin{subfigure}[t]{0.23\textwidth}
    \centering
    \includegraphics[width=\textwidth]{ 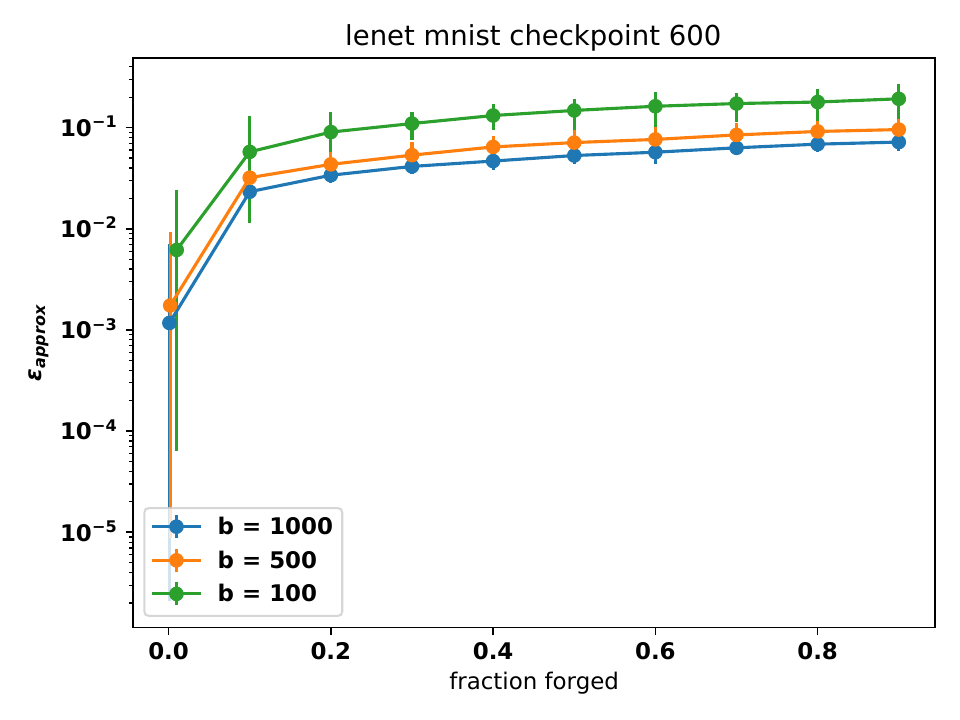 }
    \caption{LeNet/MNIST}
    \label{fig:ff-lenet}
  \end{subfigure}
  ~
  \begin{subfigure}[t]{0.23\textwidth}
    \centering
    \includegraphics[width=\textwidth]{ 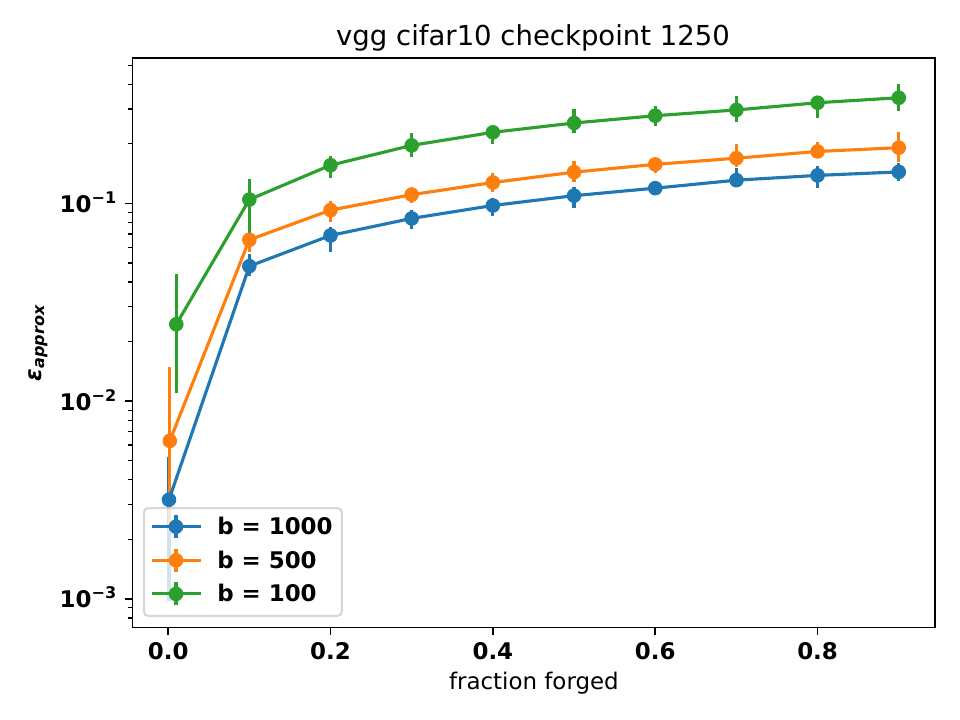 }
    \caption{VGGmini/CIFAR10}
    \label{fig:ff-vgg}
  \end{subfigure}
  \caption{The average approximation error of Thudi et al.'s
    attack~\cite{thudi2022necessity} (over several runs), as well as
    the min and max observed approximation error for different batch
    sizes and increasing forging fractions. Performance degrades
    quickly with increased forging fraction. Additionally, there is
    high variance when forging just one example. In both settings, we
    forge the middle checkpoint of the execution trace.}
  \label{fig:forging-frac}
\end{figure}

\subsubsection{\tbf{Forging performance is inconsistent across training}}

Data forging involves replacing all occurences of a given set training
examples \(U\) in an execution trace with distinct examples that
produce a similar gradient. Consequently, we postulate that the
location of the mini-batches \(B_i\) such that
\(B_i \cap U \ne \emptyset\) can affect the peformance of these data forging
attacks. Given a trace that captures \(E\) total epochs, mini-batches
that contain a given example \((\tbf{x},\tbf{y})\) will occur \(E\)
times roughly uniformly across the execution trace. As a result, we
must evaluate the performance of data forging attacks across training
as every occurence of \((\tbf{x}, \tbf{y})\) in the trace must be
forged. We run the following experiment to evaluate this:

\begin{enumerate}
\item For a given execution trace's training trajectory
  \(\theta_0,\theta_1,\dots,\theta_{T}\), we sample a true mini-batch
  \(B\) and calculate its model update \(g(\theta_i,B)\) across the
  trajectory.
\item We then forge mini-batch \(B\) into \(B'\) across training at
  each checkpoint and measure the approximation error.
\item We repeat this several times for different sampled mini-batches
  in order to study the variance.
\end{enumerate}

Firstly, we find that when the forging fraction is \(1\),
approximation errors are at least 4 orders of magnitude larger than
our largest observed reproduction errors for every model and dataset
setup we consider for both \texttt{float32} and \texttt{float16}
formats. Approximation errors were closer to reproduction errors only
for the \texttt{bfloat16} mixed precision regime, however they were
still 2 orders of magnitude larger. In Figures
~\ref{fig:fac-fp32-ff1},~\ref{fig:fac-fp16-ff1}, and
\ref{fig:fac-bf16-ff1}, we observe that the approximation errors are
consistently orders of magnitude larger than the largest observed
reproduction errors (shown as the red dotted line).

Now consider the case where the forging fraction is
\(\frac{1}{b}\). In Figures \ref{fig:fac-fp32-ff1-2000},
\ref{fig:fac-fp16-ff1-100}, \ref{fig:fac-bf16-ff1-100} we plot the
approximation errors for a given training example
\((\tbf{x},\tbf{y})\) at uniform intervals across the execution trace
(each line represents a different example). In this setting, we see
that approximation errors can vary highly depending on both the
example being forged and the position in the execution trace where it
is being forged. For \texttt{bfloat} mixed precision training and a
forging fraction of \(1/100\), we find that approximation errors of
data forging attacks are closest to reproduction errors, being either
the same magnitude or a single order of magnitude larger when forging
single examples (see Fig. \ref{fig:fac-bf16-ff1-100}). However, the
performance of the attack is not consistent across training, as seen
in Figures \ref{fig:fac-fp32-ff1-2000}, and
\ref{fig:fac-fp16-ff1-100}, where approximation errors can vary from
being only a single order of magnitude larger to nearly 5 orders of
magnitude larger (see Fig. \ref{lenet-mnist-ff-1-2000-fp32}).

\begin{figure}[htb]
  \centering
  \begin{subfigure}[b]{0.23\textwidth}
    \centering
    \includegraphics[width=\textwidth]{ 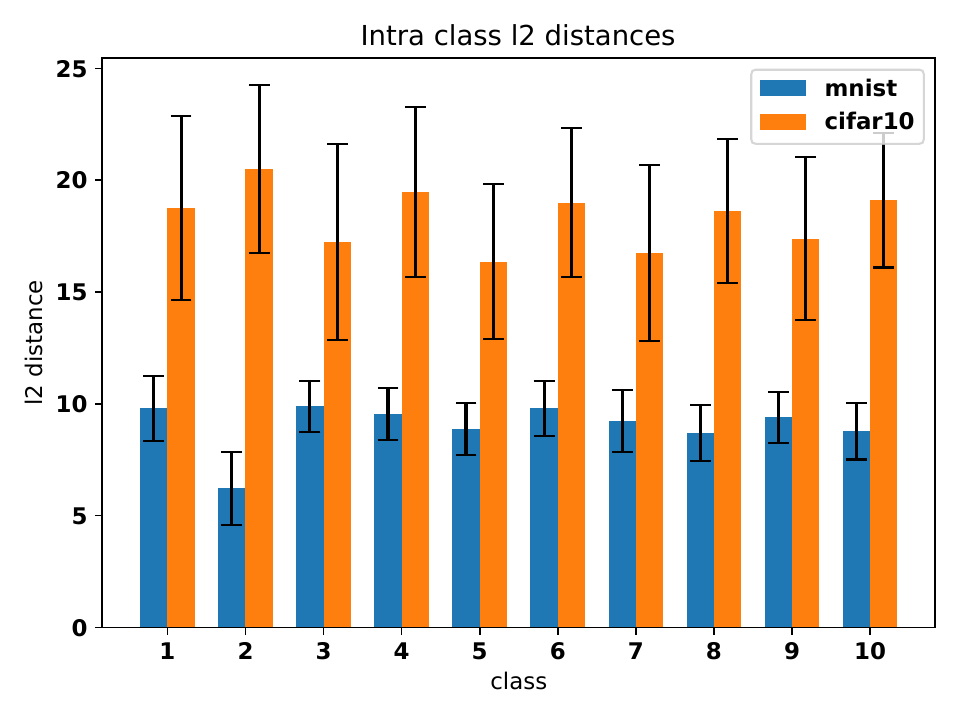 }
    \caption{}
    \label{fig:mnist-cifar-var}
  \end{subfigure}
  ~
  \begin{subfigure}[b]{0.23\textwidth}
    \centering
    \includegraphics[width=\textwidth]{ 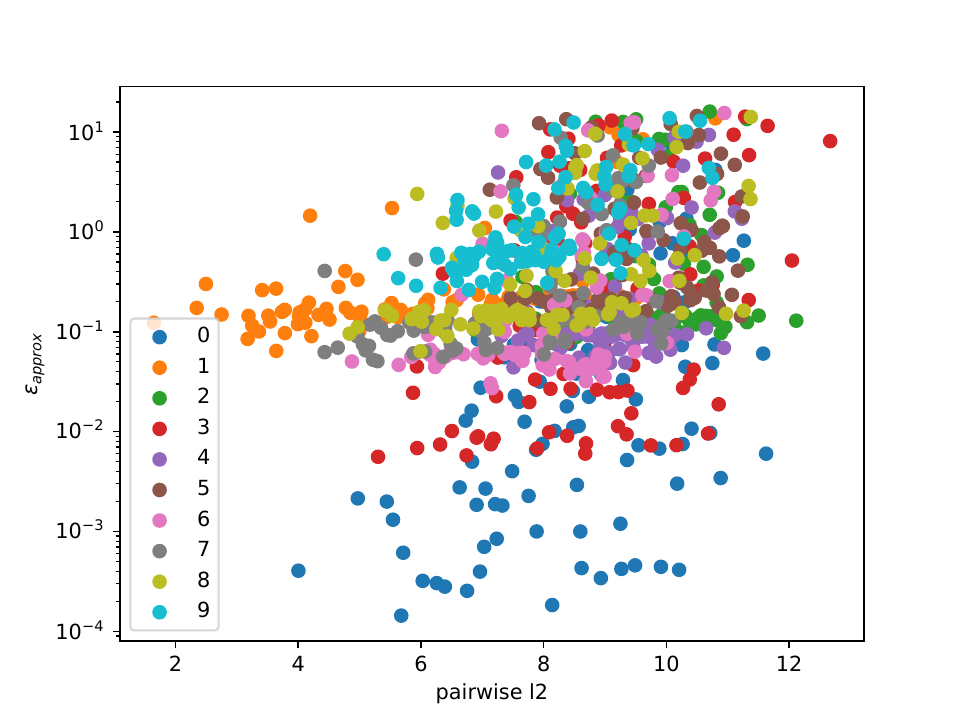 }
    \caption{}
    \label{fig:pairwise-mnist}
  \end{subfigure}    
  \caption{While MNIST has lower intra-class pairwise \(\ell_2\) variance
    than CIFAR10 (Fig \ref{fig:mnist-cifar-var}), we do not find a
    correlation between pairwise distance and approximation error (Fig
    \ref{fig:pairwise-mnist}). }
  \label{fig:dv}

\end{figure}

One possible explanation regarding the difference in forging
performance is the similarity of examples within the MNIST dataset,
however, Figure \ref{fig:dv} highlights that there is no observed
correlation between pairwise \(\ell_2\) distance and approximation
error. Instead, we highlight the relationship between approximation
error at a given checkpoint in training and the loss of the single
example being forged at that same checkpoint. Figures
\ref{lc-lenet-mnist-e-approx},\ref{lc-lenet-mnist-loss} show for the
LeNet/MNIST setup the approximation errors of data forging and the
loss of the particular example being forged (each colour in both
figures corresponds to the same training example). In these two
figures we can see a direct correlation between the loss of the model
on the particular example \((\tbf{x},\tbf{y})\) being forged, and the
approximation error. The smaller the loss, the smaller the norm of the
example gradient
\(\nabla_{\theta}\ell(f_{\theta}(\tbf{x}), \tbf{y})\), the less the example contributes
to the average gradient \(\nabla_{\theta}\mathcal{L}(B)\). As a result, when forging such
an example, data forging attacks find another example with a similarly
small gradient norm and replace it. Since both gradient norms are
small, replacing one with the other does not greatly affect the
direction or norm of the average gradient, resulting in small
approximation errors. However this is not the case when the norm of
the example gradient is large, such as early in training. Data forging
attacks cannot find a suitable replacement; resulting in larger
approximation errors. We find this pattern holds across model
architectures as well, see Figures \ref{lc-trans-imdb-e-approx}, and
\ref{lc-trans-imdb-loss} for similar behaviour for a transfomer.

\begin{figure*}[htb]
  \centering
  \begin{subfigure}[b]{0.23\textwidth}
    \centering
    \includegraphics[width=\textwidth]{ 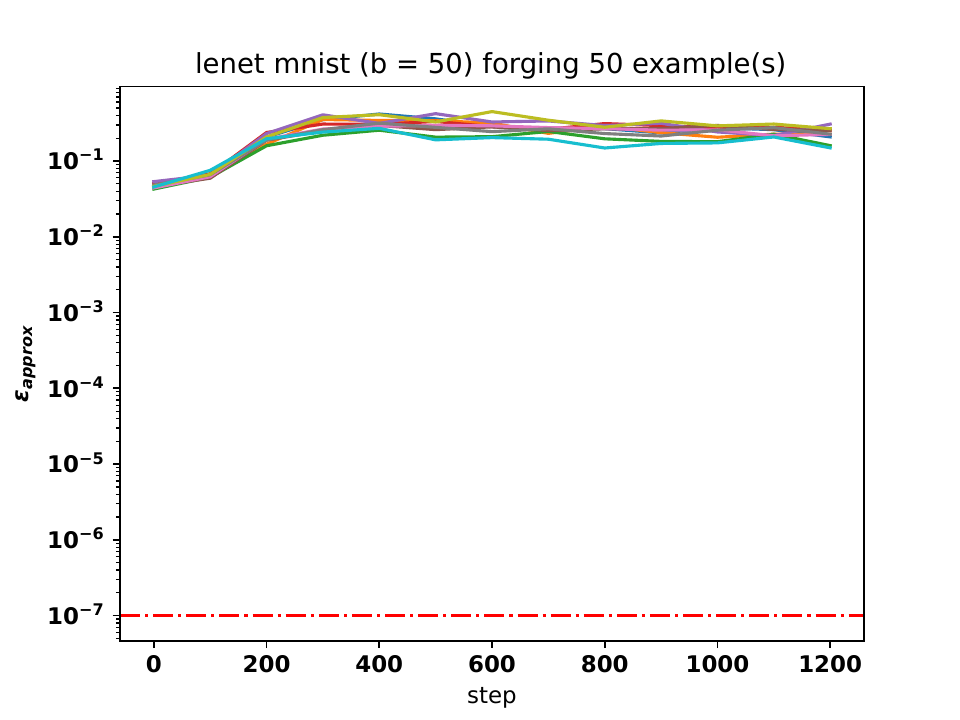 }
    \caption{LeNet/MNIST}
    \label{lenet-mnist-fp32-ff1}
  \end{subfigure}  
  ~
  \begin{subfigure}[b]{0.23\textwidth}
    \centering
    \includegraphics[width=\textwidth]{ 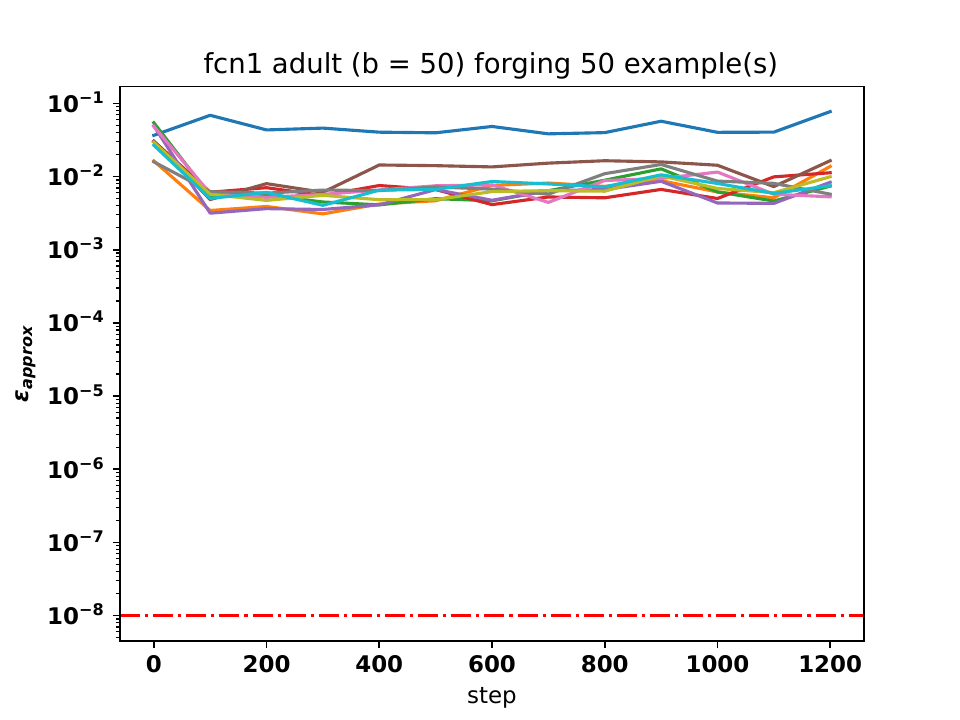 }
    \caption{FCN/Adult}
    \label{fcn-adult-fp32--ff1}
  \end{subfigure}
  ~
  \begin{subfigure}[b]{0.23\textwidth}
    \centering
    \includegraphics[width=\textwidth]{ 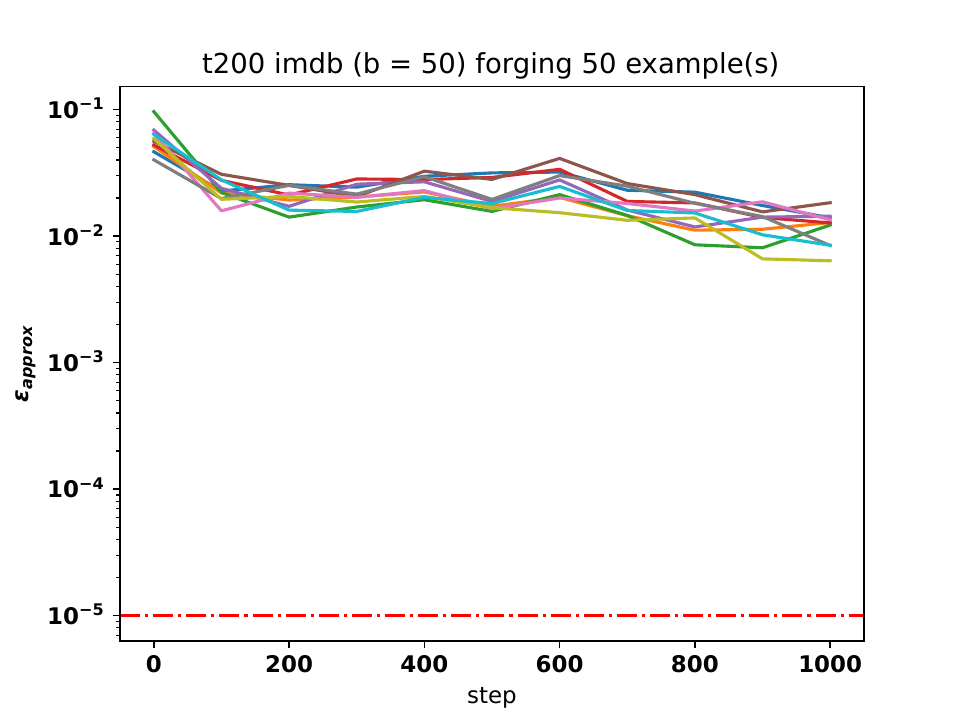 }
    \caption{Trans/IMDB}
    \label{trans-imdb-fp32-ff1}
  \end{subfigure}  
  ~
  \begin{subfigure}[b]{0.23\textwidth}
    \centering
    \includegraphics[width=\textwidth]{ 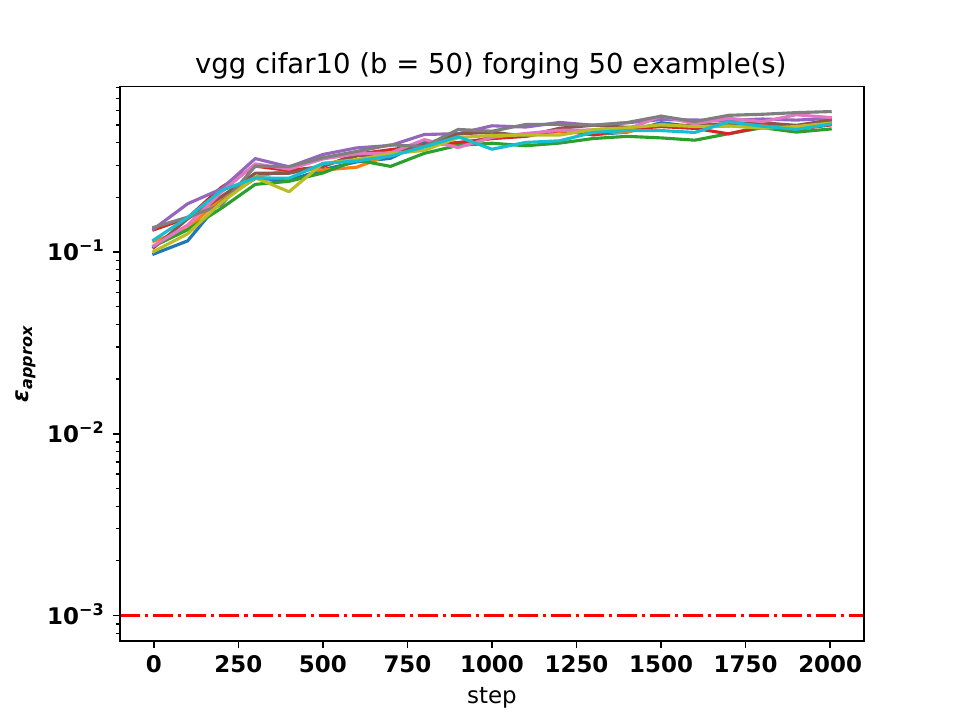 }
    \caption{VGGmini/CIFAR10}
    \label{vgg-cifar-fp32-ff1}
  \end{subfigure}  
  \caption{Data forging approximation errors across stages of training
    where the forging fraction is \(1\) for \texttt{float32}
    training. Each line corresponds to a different true mini-batch
    \(B\), and the y-axis reports the approximation error when trying
    to forge \(B\) at the given step on the x-axis.}
  \label{fig:fac-fp32-ff1}
\end{figure*}

\begin{figure*}[htb]
  \centering
  \begin{subfigure}[b]{0.23\textwidth}
    \centering
    \includegraphics[width=\textwidth]{ 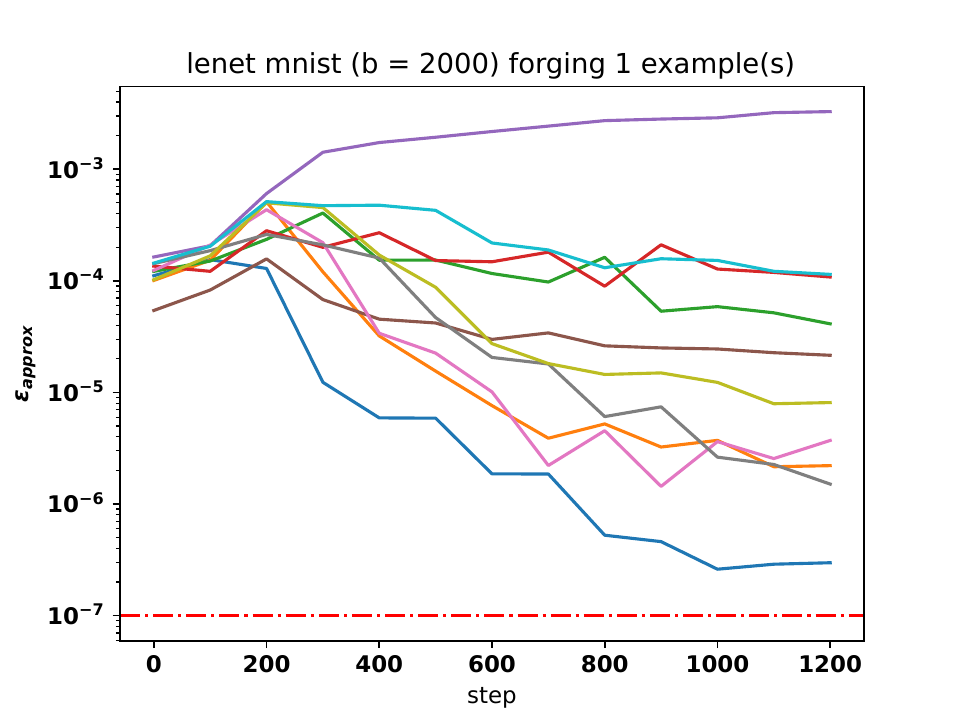 }
    \caption{LeNet/MNIST}
    \label{lenet-mnist-ff-1-2000-fp32}
  \end{subfigure}  
  ~
  \begin{subfigure}[b]{0.23\textwidth}
    \centering
    \includegraphics[width=\textwidth]{ 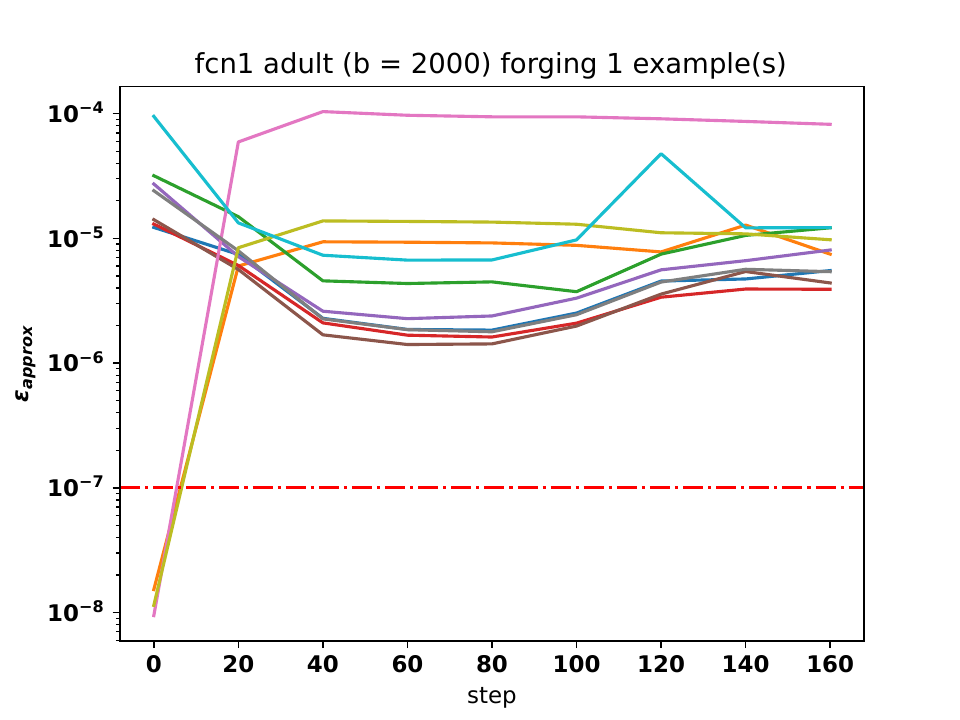 } 
    \caption{FCN/Adult}
    \label{}
  \end{subfigure}
  ~
  \begin{subfigure}[b]{0.23\textwidth}
    \centering
    \includegraphics[width=\textwidth]{ 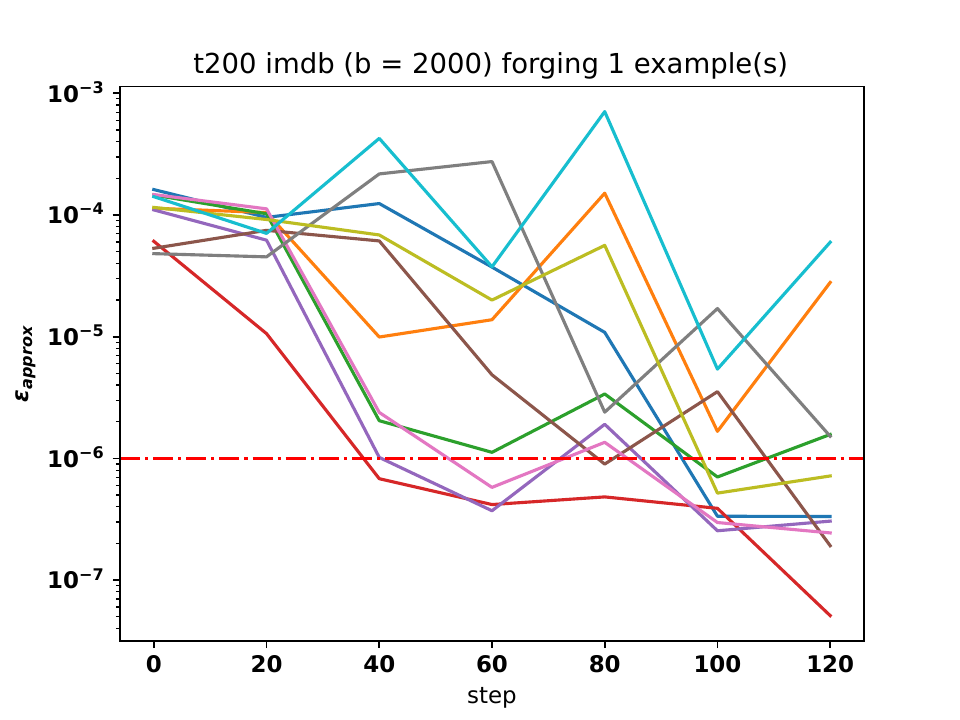 }
    \caption{Trans/IMDB}
    \label{}
  \end{subfigure}  
  ~
  \begin{subfigure}[b]{0.23\textwidth}
    \centering
    \includegraphics[width=\textwidth]{ 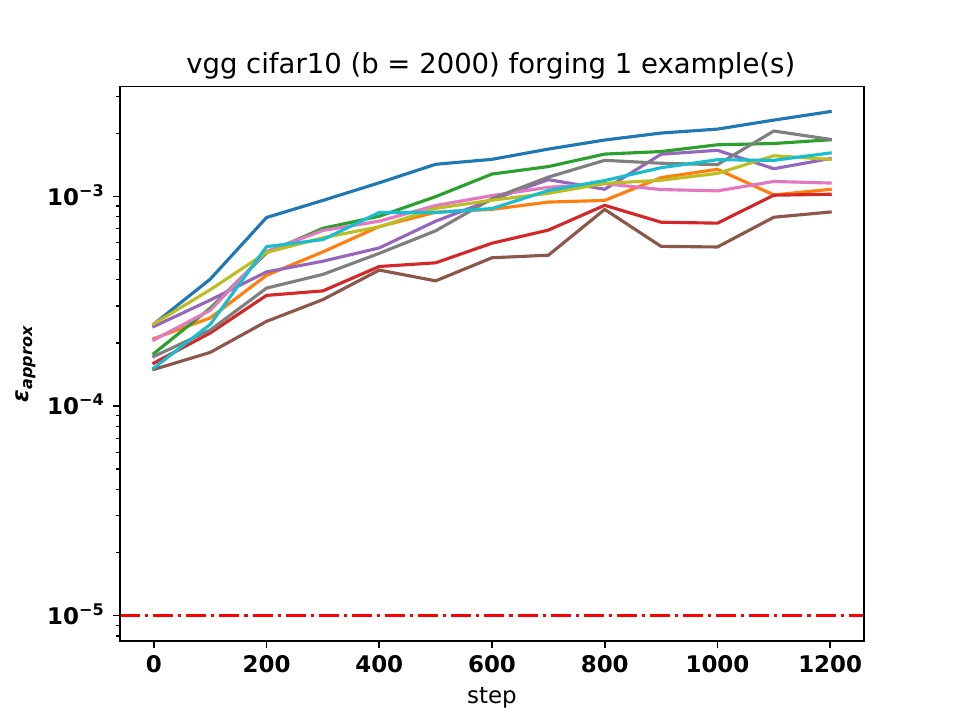 }
    \caption{VGGmini/CIFAR10}
    \label{}
  \end{subfigure}  
  \caption{Data forging approximation errors across stages of training
    where the forging fraction is \(1/2000\) for \texttt{float32}
    training.}
  \label{fig:fac-fp32-ff1-2000}
\end{figure*}

\begin{figure*}[htb]
  \centering
  \begin{subfigure}[b]{0.23\textwidth}
    \centering
    \includegraphics[width=\textwidth]{ 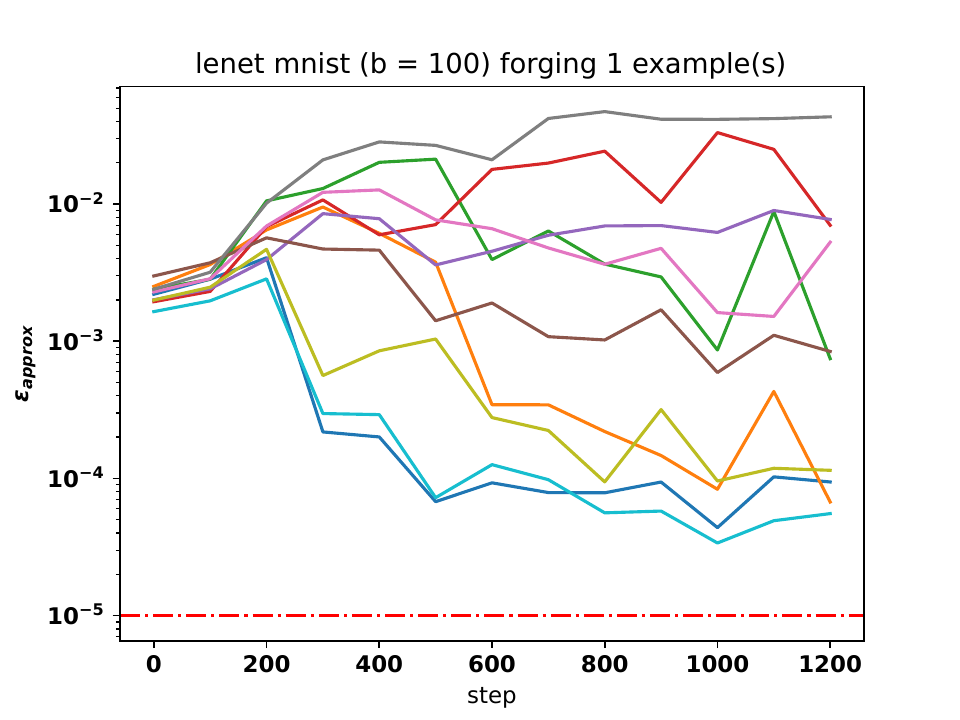 }
    \caption{LeNet/MNIST}
    \label{}
  \end{subfigure}  
  ~
  \begin{subfigure}[b]{0.23\textwidth}
    \centering
    \includegraphics[width=\textwidth]{ 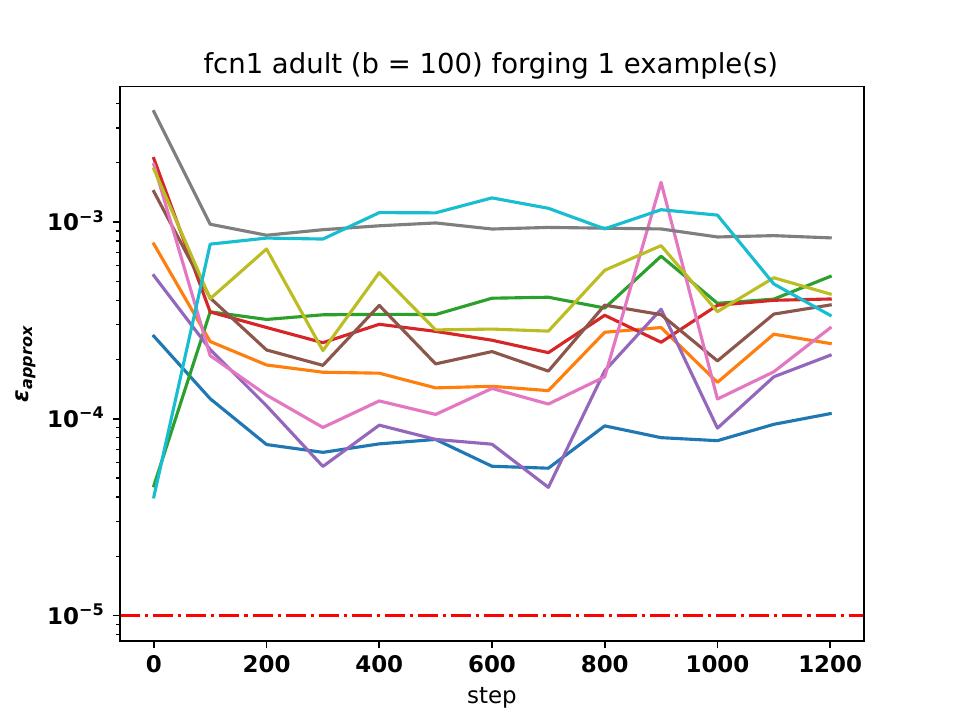 }
    \caption{FCN/Adult}
    \label{}
  \end{subfigure}
  ~
  \begin{subfigure}[b]{0.23\textwidth}
    \centering
    \includegraphics[width=\textwidth]{ 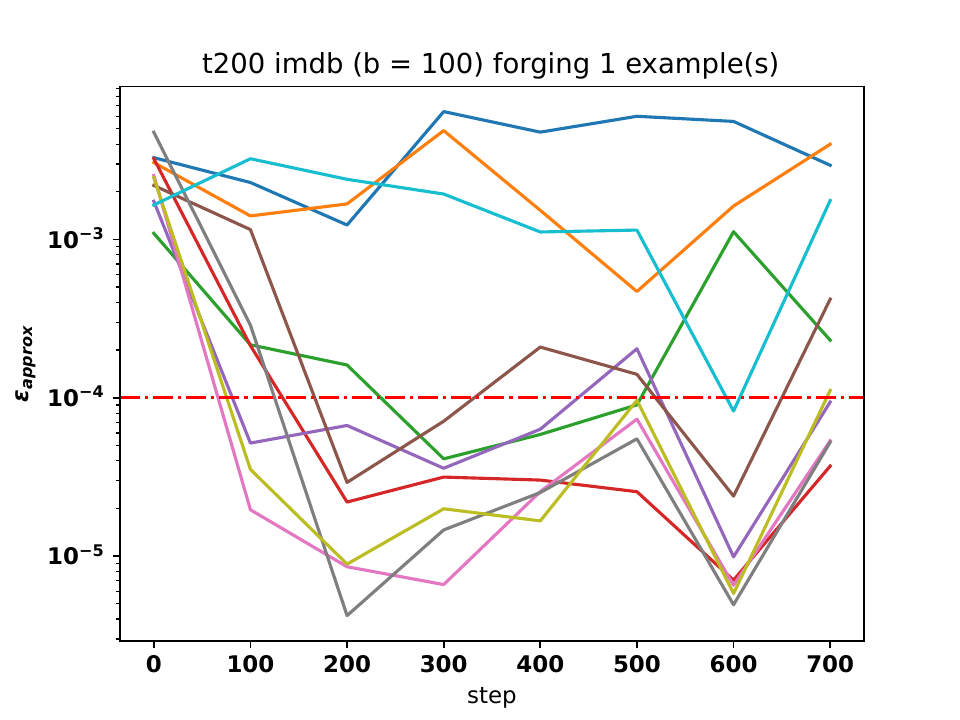 }
    \caption{Trans/IMDB}
    \label{}
  \end{subfigure}  
  ~
  \begin{subfigure}[b]{0.23\textwidth}
    \centering
    \includegraphics[width=\textwidth]{ 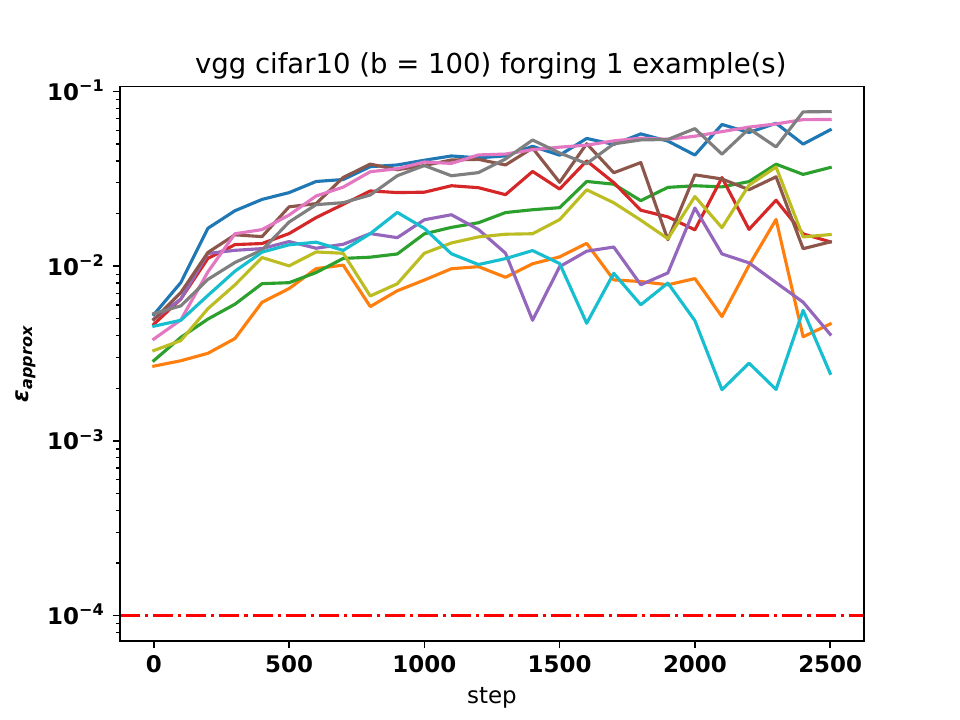 }
    \caption{VGGmini/CIFAR10}
    \label{}
  \end{subfigure}  
  \caption{Data forging approximation errors across stages of training
    where the forging fraction is \(1/100\) for \texttt{float16}
    training.}
  \label{fig:fac-fp16-ff1-100}
\end{figure*}

\begin{figure*}[htb]
  \centering
  \begin{subfigure}[b]{0.23\textwidth}
    \centering
    \includegraphics[width=\textwidth]{ 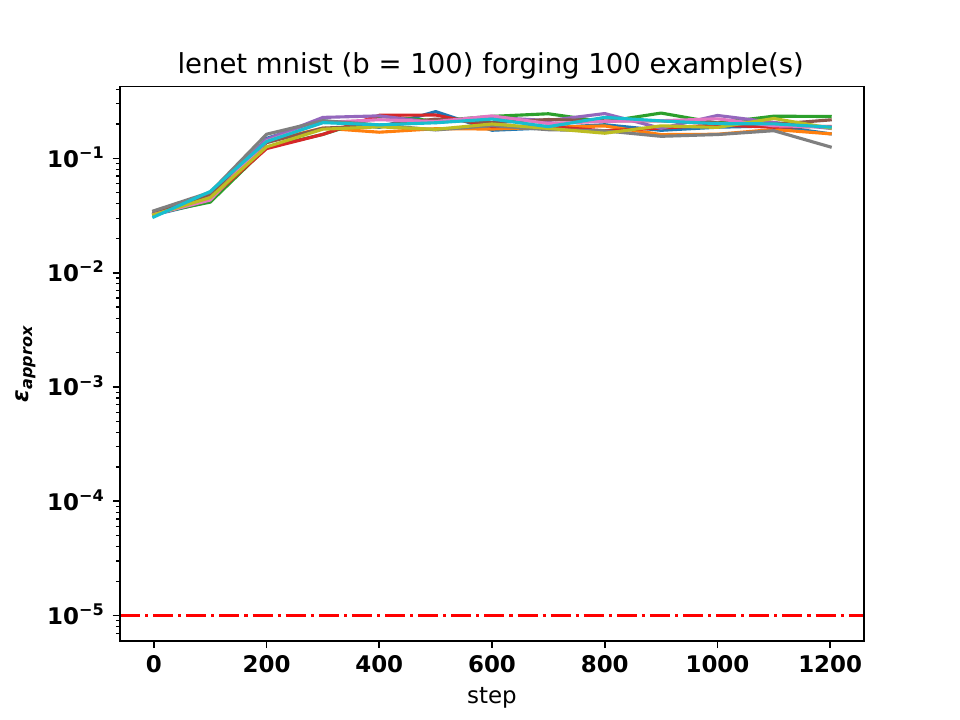 }
    \caption{LeNet/MNIST}
    \label{}
  \end{subfigure}  
  ~
  \begin{subfigure}[b]{0.23\textwidth}
    \centering
    \includegraphics[width=\textwidth]{ 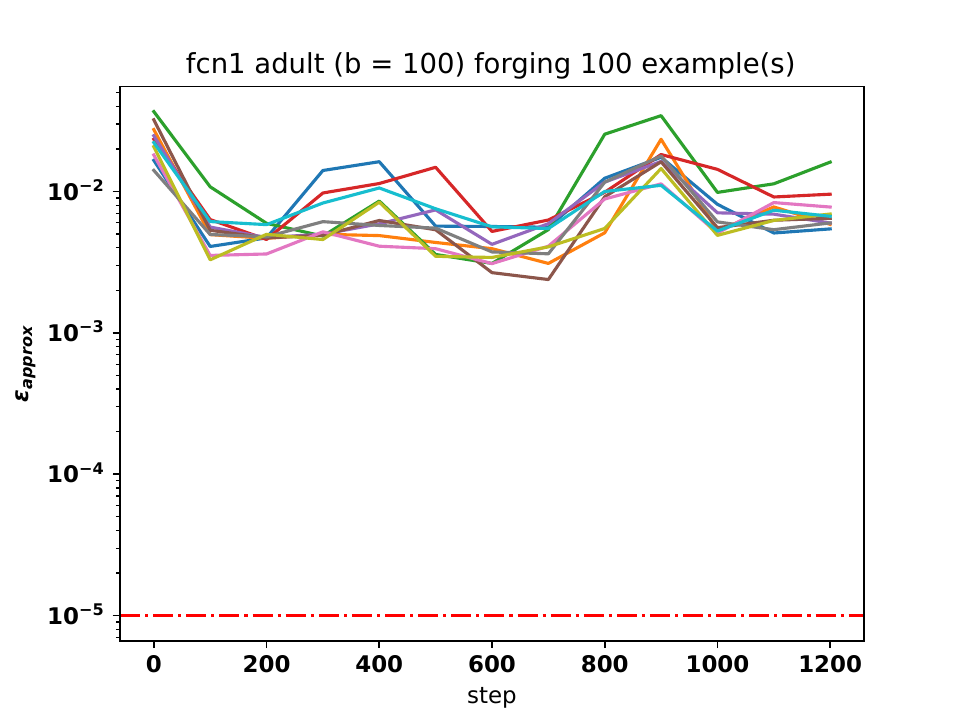 }
    \caption{FCN/Adult}
    \label{}
  \end{subfigure}
  ~
  \begin{subfigure}[b]{0.23\textwidth}
    \centering
    \includegraphics[width=\textwidth]{ 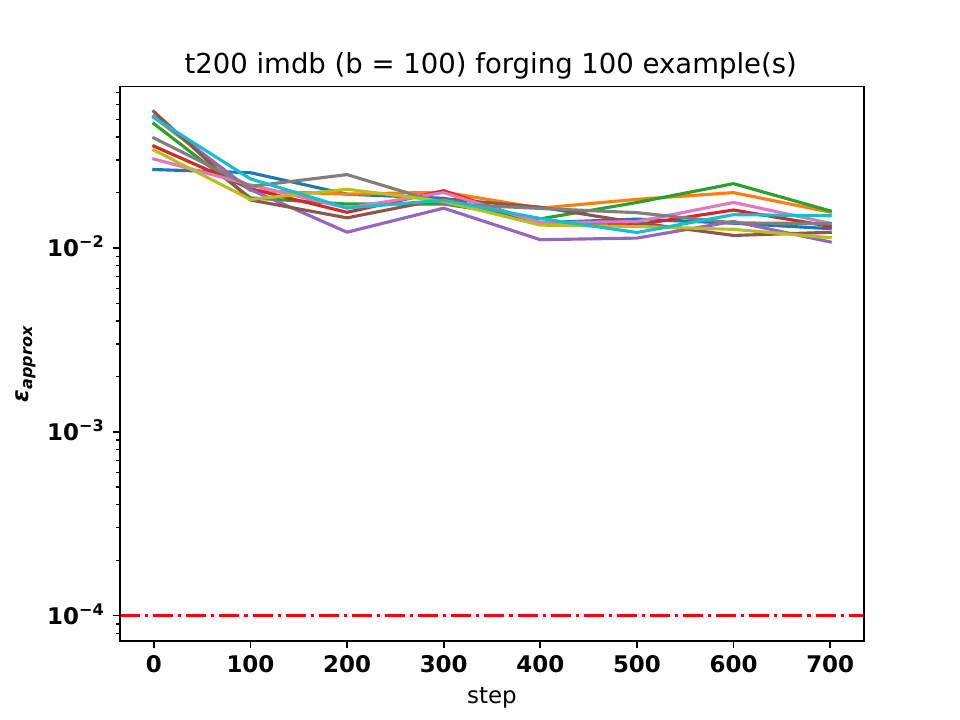 }
    \caption{Trans/IMDB}
    \label{}
  \end{subfigure}  
  ~
  \begin{subfigure}[b]{0.23\textwidth}
    \centering
    \includegraphics[width=\textwidth]{ 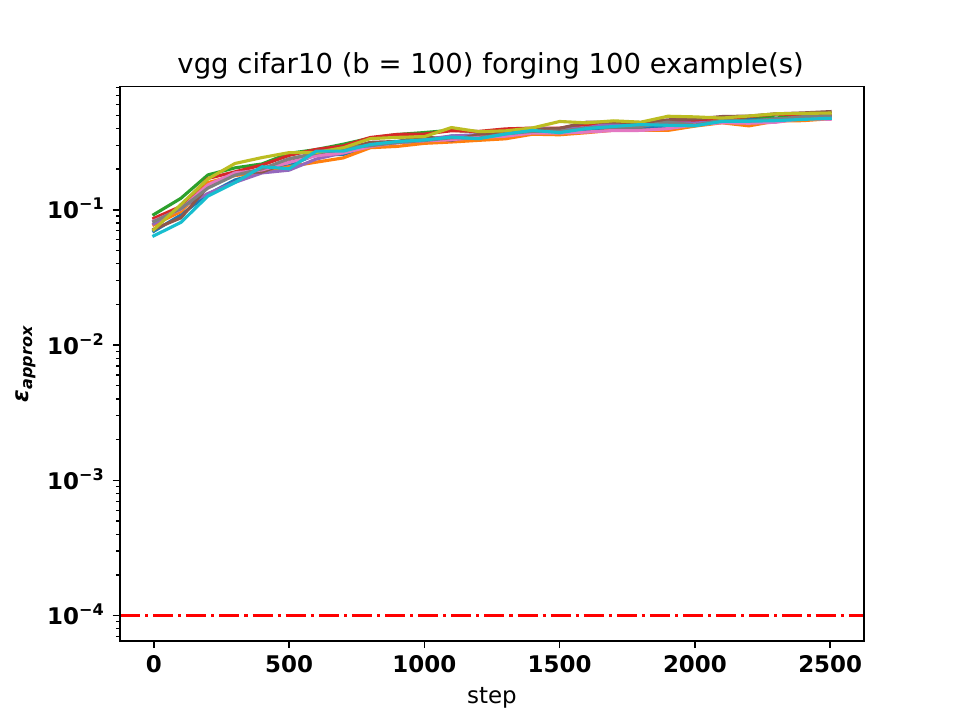 }
    \caption{VGGmini/CIFAR10}
    \label{}
  \end{subfigure}  
  \caption{Data forging approximation errors across stages of training
    where the forging fraction is \(1\) for \texttt{float16}
    training.}
  \label{fig:fac-fp16-ff1}
\end{figure*}

\begin{figure*}[htb]
  \centering
  \begin{subfigure}[b]{0.23\textwidth}
    \centering
    \includegraphics[width=\textwidth]{ 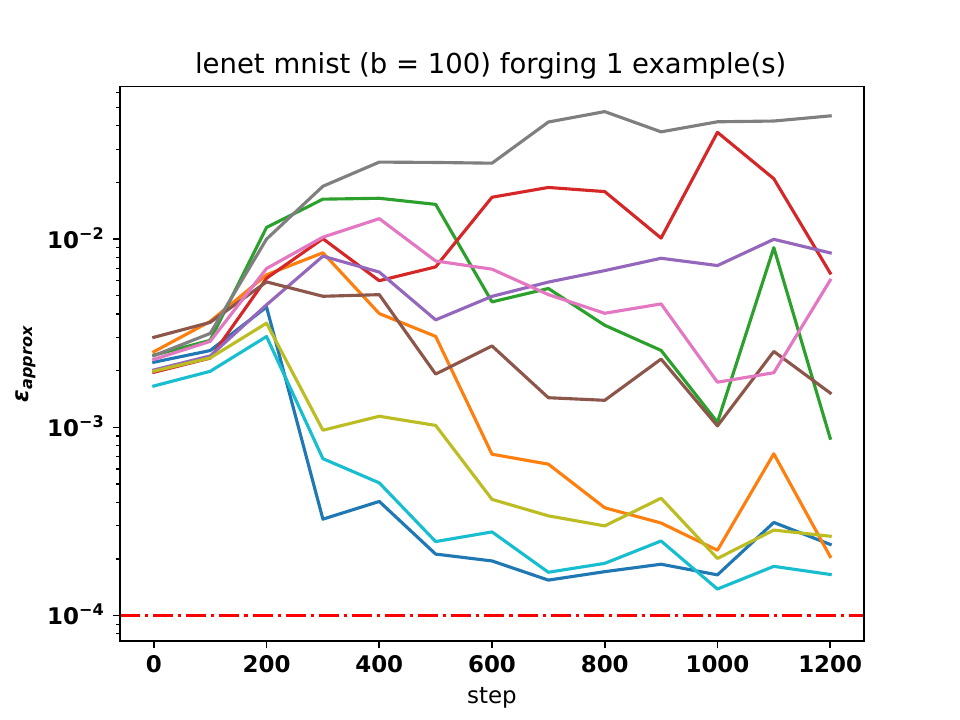 }
    \caption{LeNet/MNIST}
    \label{}
  \end{subfigure}  
  ~
  \begin{subfigure}[b]{0.23\textwidth}
    \centering
    \includegraphics[width=\textwidth]{ 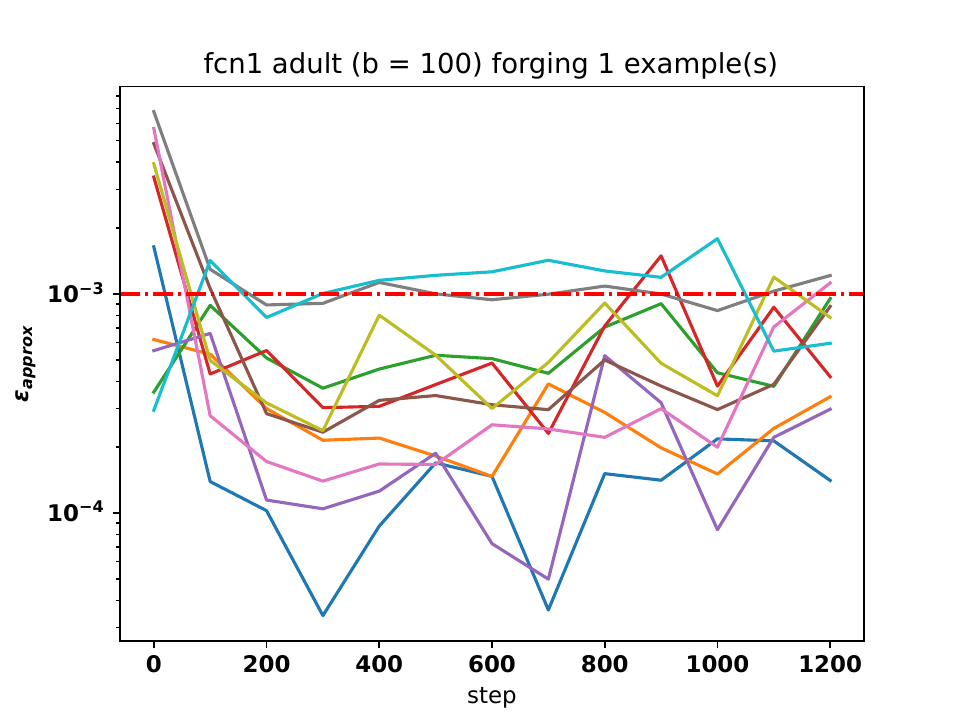 }
    \caption{FCN/Adult}
    \label{}
  \end{subfigure}
  ~
  \begin{subfigure}[b]{0.23\textwidth}
    \centering
    \includegraphics[width=\textwidth]{ 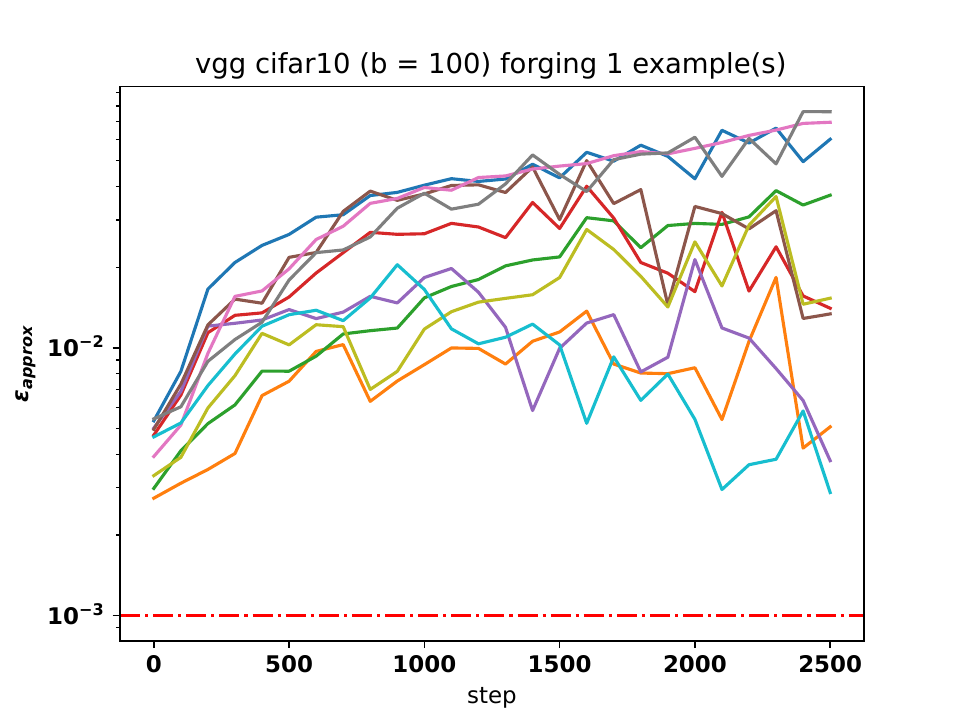 }
    \caption{VGGmini/CIFAR10}
    \label{}
  \end{subfigure}  
  \caption{Data forging approximation errors across stages of training
    where the forging fraction is \(1/100\) for \texttt{bfloat16}
    training.}
  \label{fig:fac-bf16-ff1-100}
\end{figure*}

\begin{figure*}[htb]
  \centering
  \begin{subfigure}[b]{0.23\textwidth}
    \centering
    \includegraphics[width=\textwidth]{ 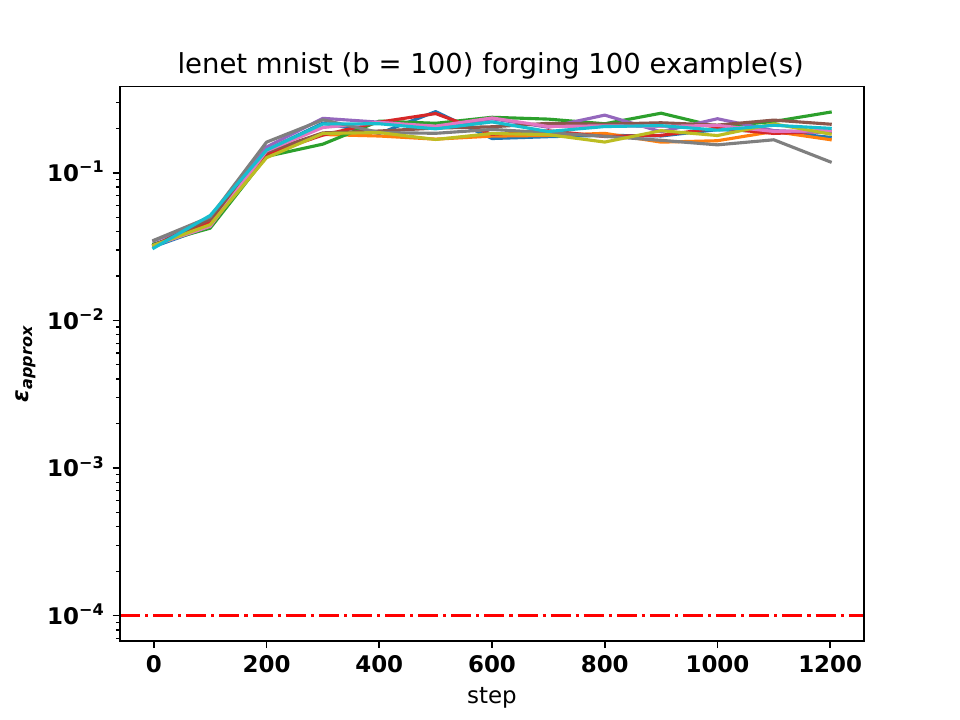 }
    \caption{LeNet/MNIST}
    \label{}
  \end{subfigure}  
  ~
  \begin{subfigure}[b]{0.23\textwidth}
    \centering
    \includegraphics[width=\textwidth]{ 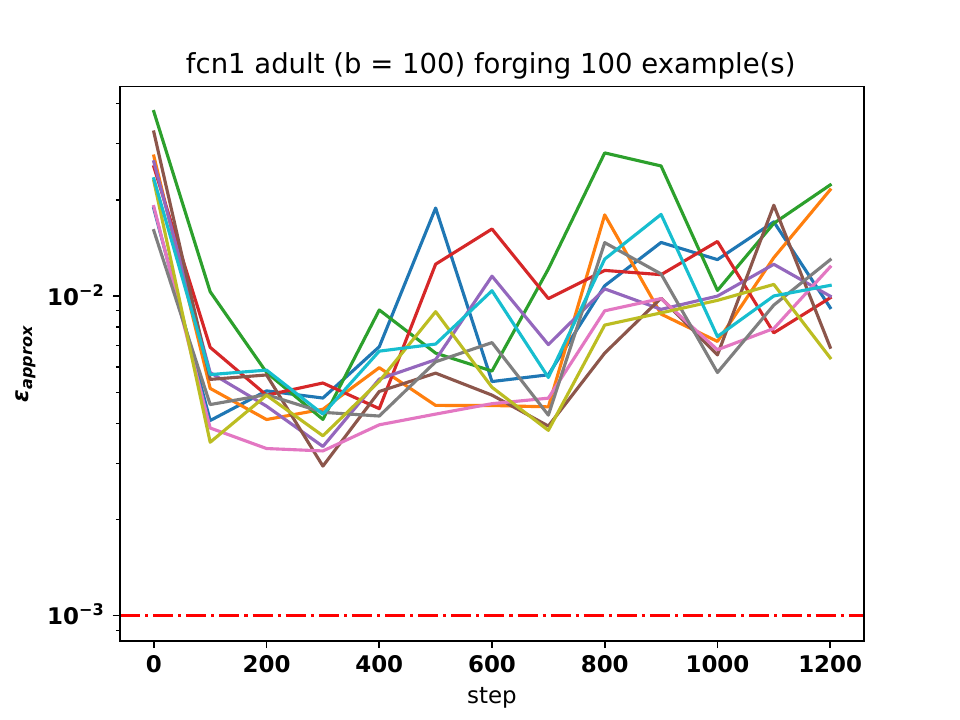 }
    \caption{FCN/Adult}
    \label{}
  \end{subfigure}
  ~
  \begin{subfigure}[b]{0.23\textwidth}
    \centering
    \includegraphics[width=\textwidth]{ 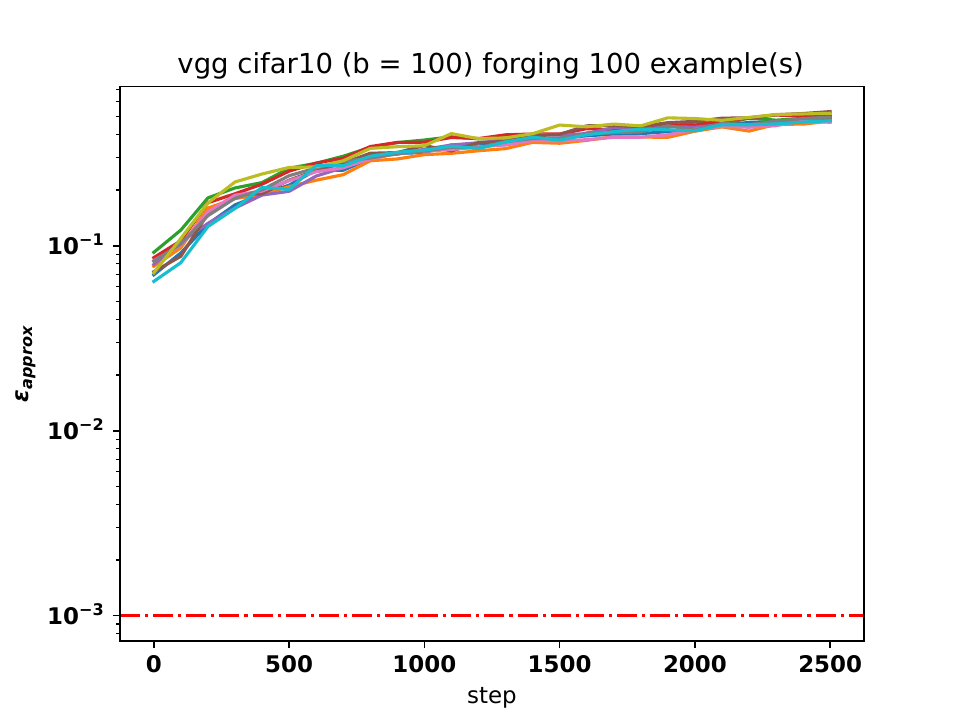 }
    \caption{VGGmini/CIFAR10}
    \label{}
  \end{subfigure}  
  \caption{Data forging approximation errors across stages of training
    where the forging fraction is \(1\) for \texttt{bfloat16}
    training.}
  \label{fig:fac-bf16-ff1}
\end{figure*}

\begin{figure}[htb]
  \centering
  \begin{subfigure}[b]{0.23\textwidth}
    \centering
    \includegraphics[width=\textwidth]{ ./to_forging_section/20250514T145003--measure-attack-across-checkpoints-lenet-mnist__b2000_e40_k1_seed2024__f1_thudi/approx_errors.pdf }
    \caption{LeNet/MNIST}
    \label{lc-lenet-mnist-e-approx}
  \end{subfigure}  
  ~
  \begin{subfigure}[b]{0.23\textwidth}
    \centering
    \includegraphics[width=\textwidth]{ 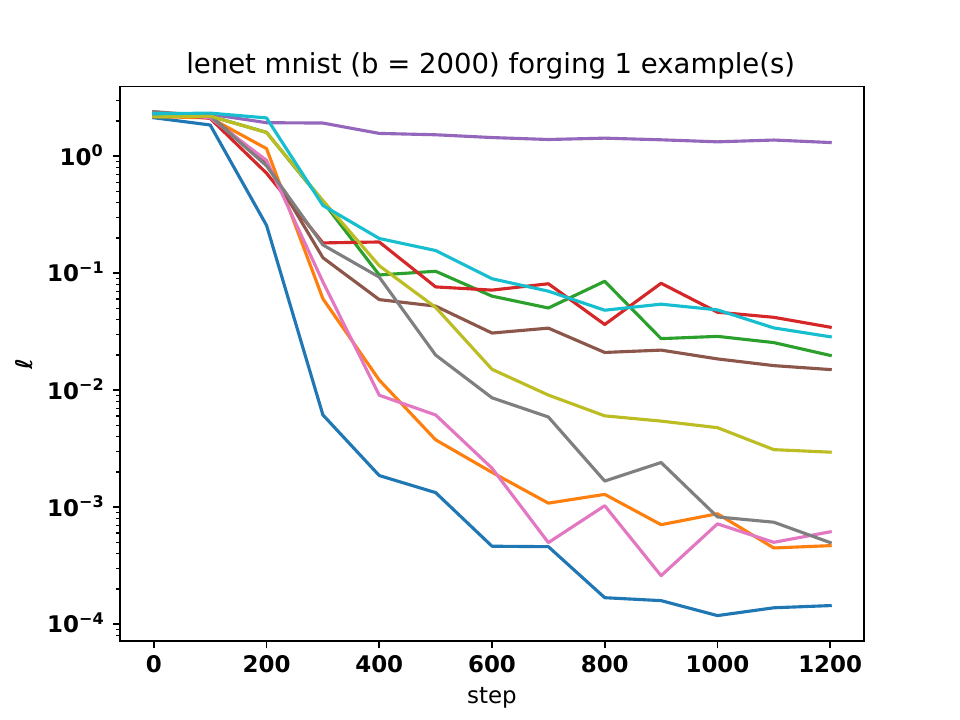 }
    \caption{LeNet/MNIST}
    \label{lc-lenet-mnist-loss}
  \end{subfigure}  

  \begin{subfigure}[b]{0.23\textwidth}
    \centering
    \includegraphics[width=\textwidth]{ ./to_forging_section/20250520T092714--measure-attack-across-checkpoints-t200-imdb__b2000_e10_k1_seed2024__f1_thudi/approx_errors.pdf }
    \caption{Trans/IMDB}
    \label{lc-trans-imdb-e-approx}
  \end{subfigure}  
  ~
  \begin{subfigure}[b]{0.23\textwidth}
    \centering
    \includegraphics[width=\textwidth]{ 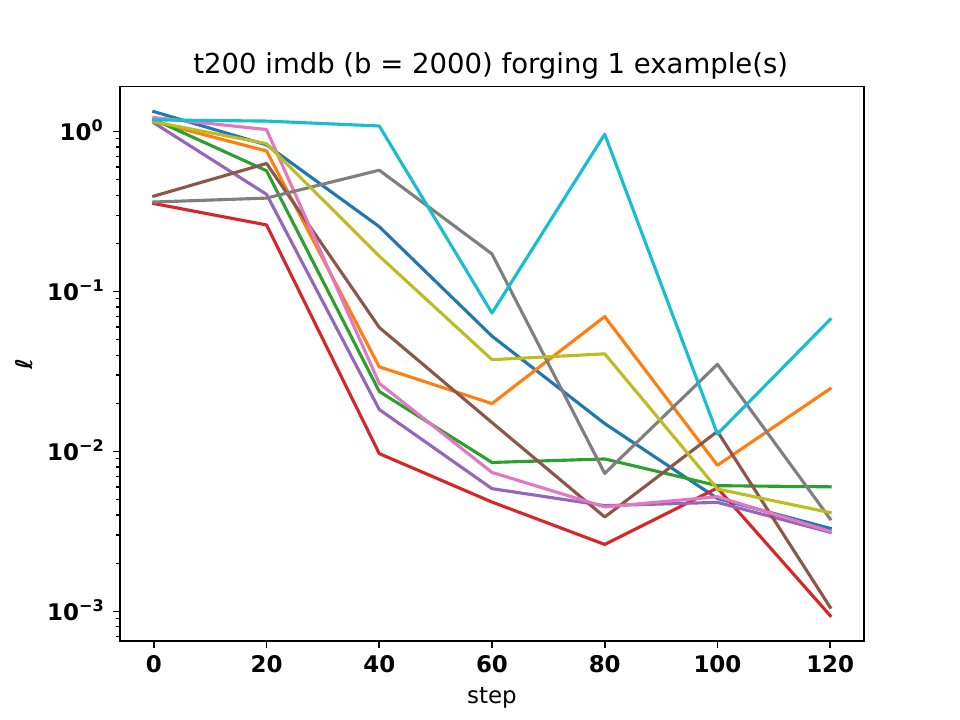 }
    \caption{Trans/IMDB}
    \label{lc-trans-imdb-loss}
  \end{subfigure}
  \caption{The correlation between the approximation error when
    forging a single example across an execution trace and the loss of
    that example at a given checkpoint. }
  \label{lc-lm-ti}
\end{figure}

\subsection{Is Data Forging Stealthy?}

From our experimentation in Section \ref{sec:what-magn-benign}, we
found that the model type, batch size, floating point precision, and
GPU hardware can affect the magnitude of benign reproduction errors,
and in Section \ref{sec:when-data-forging}, we evaluated the
performance data forging attacks. \tbf{From our evaluation of the
  existing data forging attacks, they cannot be classed as
  stealthy}. While factors such as forging fraction, the stage of
training, and floating point precision can affect the performance of
these attacks, overall they are not stealthy. In the case where the
forging fraction is \(1\), data forging attacks consistently produce
approximation errors that are several orders of magnitude larger than
benign reproduction errors. When the forging fraction is \(1/b\), we
found that the performance is inconsistent across training. Depending
on the value of the loss function for the particular example that is
being forged, the approximation error can be close to reproduction
error or be several orders of magnitude larger.

\subsection{Discussion: Choosing An Error Threshold In Practice is Hard}
\label{sec:choosing-an-error}
From our experimentation, we have seen that existing attacks are not
stealthy when the error threshold \(\epsilon\) is chosen to be no larger than
benign reproduction errors. When error thresholds are accurately
chosen via thorough experimentation to determine the magnitude of
reproduction errors, we have seen that data forging attacks are
consistently not stealthy. As a result, we contend that such a method
for determining the model's provenance \emph{i.e.}, its true training
data, is adequate. Existing forging attacks cannot produce forged
execution traces that can fool an informed verifier.

However the approach is not perfect. Determining accurate error
thresholds value can be tricky when it is not possible to perform
cross hardware recomputation. Consider the following scenario that
highlights the difficulty of correctly setting \(\epsilon\) in practice: a
verifier intends to verify a LeNet/MNIST \((b=100)\) execution trace
on a RTX4090 GPU. They set \(\epsilon = 10^{-8}\) for their error threshold
as this is their observed level of reproduction error on their
hardware\footnote{See the corresponding plot in Figure \ref{fig:bs}
  for LeNet/MNIST (\(b=100\)).}. Now suppose that the model owner is
honest and trained their model using a Tesla V100 GPU. From our
experimentation, we found that the reproduction error of honest traces
generated using a V100 and recomputed on a RTX4080 can be as large as
\(10^{-6}\) (See Table \ref{tab:ch-lenetb100_determ}). As a result,
even the honest trace would not pass verification if
\(\epsilon = 10^{-8}\). One of the challenges of building robust verification
schemes is that, as we have seen, reproduction errors can vary by
several orders of magnitude depending on the hardware that the
checkpoints were generated on, as well as the hardware they are
reproduced on.

Within the data forging literature, there is a difference of opinion
on what is an acceptable choice of \(\epsilon\). Baluta et
al. \cite{baluta2023unforgeability} consider only the scenario of a
strict verifier that always selects \(\epsilon = 0\). They assert that given
\emph{(i)} the initial random seed(s) used during training,
\emph{(ii)} the original software, and \emph{(iii)} the original
hardware the computations were performed on, there should be zero
error between recomputations of any given execution trace. They claim
that since the entire computation sequence is deterministic, error at
any step of reproduction of the execution trace must therefore point
to foul play. While it is true that reproduction error can be
eliminated given these conditions\footnote{See Table
  \ref{tab:ch-lenetb100_determ}, where we observed zero reproduction
  error when the original random seeds, hardware, and software were
  used for the recomputation.} and may be necessary in high stakes
applications such as medical machine learning models, we argue that
always demanding \(\epsilon = 0\) is unrealistic in practice, particularly
when the original hardware is not available.

Due to the existence of reproduction error, prior data forging works
\cite{kong2023can,thudi2022necessity,zhangverification} have
recognised the practical necessity of a non-zero threshold during
verification, and consider the \(\epsilon > 0\) scenario. While their attacks
produce non-zero approximation errors, the question of whether a
verifier should accept the levels of approximation error generated by
their attack algorithms is largely overlooked by the prior work. They
provide no analysis of how large reproduction errors can be and
whether their attack algorithms produce approximation errors that are
comparable. In this work, we have conclusively determined that this is
not the case; reproduction errors are orders of magnitude smaller than
the approximation errors produced by existing attacks.

Verifiers are therefore met with a tradeoff when choosing their error
threshold. Seeking to not reject honest traces, and being cognisant of
the fact that there may exist other factors that affect reproduction
errors; they set their threshold to be larger than the reproduction
errors which they have observed through experimentation. However, in
doing so, they allow for data forging attacks, that would have not
been stealthy under a stricter choice of error threshold, to become
stealthy, thus exploiting the verifier's pragmatism. A similar
situation is seen in the area of Differential Privacy (DP)
\cite{dwork2006differential}. A privacy parameter \(\epsilon\) is chosen to
determine how private\footnote{Privacy in DP is quantified by how much
  the output of a function differs for inputs that differ by one
  datapoint. If they do not differ by much, then inferring whether a
  datapoint was or was not part of the input is difficult. This
  difficulty of inference is inversely proportional to the privacy
  parameter \(\epsilon\).} a given function is. Within DP, ideally
\(\epsilon < 1\), however, choosing such a small epsilon can hinder the
performance of the model. We have a similar setup here, where
\(\epsilon = 0\) gives the verifier the most confidence in their
reproductions, however this may not be practical all the time,
necessitating the choice of larger error thresholds. Large error
thresholds \(\epsilon\), larger batch sizes \(b\), mixed precision training,
and smaller forging fractions allow for more \emph{wiggle room},
increasing the likelihood of potential foul play from data
forging. 

\subsubsection{Data Protection, Differential Privacy, and Data
  Forging}
Data protection legislation such as GDPR aims to limit the potential
harm that may befall an individual from the use of their personal
information. Existing methods to do so are \emph{(i)} to limit the use
of their personal data entirely, or \emph{(ii)} to provide guarantees
that any undue harm will not result from the use of their
data. Regarding the latter, the differential privacy framework has
emerged as the state of the art for provable guarantees within machine
learning, however it is not without its tradeoffs. Namely, providing
meaningful privacy guarantees means poor model performance, resulting
in concessions in order to acheive better utility, such as relying on
``public'' pre-training data \cite{li2021large} to privately fine-tune
LLMs, or the use of empirical defences (e.g. DPSGD with large values
for epsilon) in order to thwart existing attacks
\cite{aerni2024evaluations}. The former forgoes privacy guarantees for
the ``public'' data LLMs are pre-trained on, data whose use often was
not consented to and is liable for recovery from these models
\cite{carlini2021extracting,carlini2023extracting,brown2022does}, and
the latter provides no protection against potentially stronger
attacks.

The emergence of Proof of Learning (PoL) \cite{jia2021proof}, i.e.,
the recomputation of execution traces, and further zero-knowledge
methods of PoL \cite{abbaszadeh2024zero} can be viewed as an alternate
approach to limit harm by providing proof of exactly what a model was
trained on, and more importantly, what data the model was not trained
on. However this is also not without its tradeoffs. In order to thwart
data forging, accurately chosen error thresholds are needed to prevent
the forged traces from being accepted and valid traces from being
rejected. Our work has shown this requires thorough experimentation to
determine accurately. Additionally, closed-source GPU driver code
makes this task harder. Further, proving a model's training data
requires the recomputation of every step which in some use cases may
not be possible e.g., large transformer models on web-scale
datasets. Further research into less computationally intensive schemes
is necessary \cite{fang2023proof,jia2021proof} for PoL to become more
practical, as well as the utilisation of open-source, reproducible
floating point algorithms.

\section{On the Difficulty of Exact Data Forging}
\label{sec:diff-data-forg}

Existing data forging attacks attempt to find, for a given mini-batch
\(B\), a distinct counterpart \(B'\) such that their gradients are
\emph{approximately} the same. In previous sections, we have examined
how similar the gradients of the original and forged mini-batches
produced by existing ``greedy search'' style attacks
\cite{kong2023can,thudi2022necessity,zhangverification} are, and found
that the magnitude of their approximation errors are too large,
allowing us to classify existing data forging attacks as not
stealthy. However, little work has been devoted to \emph{exact} data
forging \emph{i.e.}, developing data forging attacks that produce
forged mini-batches with the \emph{same} gradient. The related subject
of gradient inversion \cite{zhu2019dlg,pan2022exploring,yin2021see,
  zhao2020idlg, yang2023gradient,zhu2021r} asks the reverse question:
\emph{given the gradient \(\nabla_{\theta}\mathcal{L}(B)\), can the original mini-batch
  \(B\) be recovered?} This question is pertinent to the area of
Federated Learning \cite{mcmahan2017communication}, as the gradients
of sensitive information are sent among several parties; and an
analysis of the privacy impact of sharing these gradients is
required. Within data forging, we are concerned with a related
question: \emph{given mini-batch \(B\) and its gradient
  \(\nabla_{\theta}\mathcal{L}(B)\), does there exist another, distinct, mini-batch
  \(B'\) such that
  \(\nabla_{\theta}\mathcal{L}(B) = \nabla_{\theta}\mathcal{L}(B')\)?} Any error when reproducing the
gradients of such a \(B\) and \(B'\) can only ever be at most, as
large as reproduction error, resulting in stealthy data
forging. Consequently, this question is of great importance to an
adversarial model owner looking to perform stealthy data
forging.

A key limitation of prior data forging attacks is the usage of the
finite search space \(D \backslash U\) when searching for replacements. Baluta
et al. \cite{baluta2023unforgeability} analysed the models and
datasets considered by the prior work. They find that the execution
traces considered by the prior work are \emph{unforgeable} using
existing techniques \emph{i.e.}, these attacks cannot produce a forged
mini-batch with, analytically, the same gradient. However, there may
exist training examples outside \(D \backslash U\) that result in zero
error. Previous work by \cite{suliman2024data} has shown that the use
of the gradient matching optimisation algorithm pioneered by
\cite{zhu2019dlg} does not converge to solutions with sufficiently
identical gradients, thus an analytical approach to the problem of
exact data forging is required. This section is devoted to answering
the question of exact data forging \emph{i.e.}, the search for
mini-batches outside \(D\backslash U\) that result in zero error.

Our analysis suggests that while it is possible to construct a
distinct mini-batch \(B'\) that produces the same gradient,
analytically, as \(B\), finding a \(B'\) within the allowed input
space \(\mathcal{X}\) and label space \(\mathcal{Y}\) is a non trivial task. In
particular, exact data forging involves solving a given system of
linear equations, where each solution corresponds to a forged
mini-batch \(B'\). However, of the infinite, distinct, solutions that
may exist, there are currently no known efficient techniques for
finding solutions that correspond to a mini-batch that falls within
the allowed domain \(\mathcal{X} \times \mathcal{Y}\).

\subsection{Exact Data Forging}
Within the classification context, we have a model, \emph{i.e.},
parameterised function
\(f_{\theta} : \mathcal{X} \rightarrow \mathcal{Y}\) and loss function
\(\ell : \mathcal{Y} \times \mathcal{Y} \rightarrow \mathbb{R}\). Exact data forging may be stated as the problem of
finding two distinct mini-batches \(B,B' \in \mathcal{Z}^b\), where
\(\mathcal{Z} := \mathcal{X} \times \mathcal{Y}\), such that
\(\nabla_{\theta} \mathcal{L}(B) = \nabla_{\theta}\mathcal{L}(B')\), where
\(\mathcal{L}(B) = \frac{1}{b} \sum_{(\tbf{x},\tbf{y}) \in B} \ell(f_{\theta}(\tbf{x}),
\tbf{y})\). We define distinct mini-batches as follows, ensuring that
they differ in at least one training example:
\begin{defn}[Distinct mini-batches]
  Two mini-batches \(B,B' \in \mathcal{Z}^b\) are distinct if
  \(\ \exists (\tbf{x},\tbf{y}) \in B\) such that
  \(\forall (\tbf{x}',\tbf{y}') \in B', \ (\tbf{x}, \tbf{y}) \ne (\tbf{x}',
  \tbf{y}')\).
\end{defn}

In the unrestricted setting, the training examples
\((\tbf{x}, \tbf{y}) \in B\) may take any real numbered value i.e.
\(\mathcal{Z} = \mathbb{R}^d \times \mathbb{R}^n\) (\(d\) and \(n\) are the number of input features
and classes respectively). However, depending on the given machine
learning application, the domain may be restricted. Let us consider
the image classification context, as done in the prior work, where
inputs consist of pixel values and the labels are one hot
vectors. There exist a finite number of pixel values \(v\) and \(n\)
possible one hot vectors. Modern image formats have \(v = 256\), so
the domain of possible mini-batches is restricted to the set
\(\mathcal{Z} = \{ \frac{q}{255} : q \in [0..255]\}^d \times \{ u \in \{0,1\}^n :
\sum_{i=1}^n u_i= 1\}\). We see that this domain is finite, and so one
potential exact data forging approach is to enumerate every possible
batch and search exhaustively for the one that produces the required
gradient. However, this approach is computationally intractable for
non-trivial values of \(b, d\) and \(n\) with no additional guarantee
that one may be found\footnote{There exist a total of
  \({256^d \times n \choose b}\) possible mini-batches of size \(b\).} (see
Appendix \ref{sec:tbfbr-force-forg} for how exact data forging via
brute-force fails for small values of \(b, d, v\) and \(n\)). Next we
consider the task of exact data forging fully connected neural
networks first for \(b=1\), and then consider the case for \(b>1\).

\subsection{Case for \(b =1\)}

For \(b=1\) we prove that exact data forging is indeed not possible
for fully connected neural networks. A \(L\)-layer fully connected
neural network is made up of parameters that consist of weights and
biases
\(\theta = [\tbf{W}_1, \tbf{b}_1 \tbf{W}_2, \tbf{b}_2, ..., \tbf{W}_L,
\tbf{b}_L]\) where each
\(\tbf{W}_i \in \mathbb{R}^{d_{i-1} \times d_i}, \tbf{b} \in \mathbb{R}^{d_i}\). The output of
the neural network is given by
\( f_{\theta}(\tbf{x}) = \mbox{softmax}(\tbf{z}_L)\), where
\( \tbf{z}_L = \tbf{W}_L^T \tbf{a}_{L-1} + \tbf{b}_L\),
\(\tbf{a}_i = h(\tbf{z}_i)\),
\(\tbf{z}_i = \tbf{W}_i^T\tbf{a}_{i-1} + \tbf{b}_i\), and \(h\) is the
activation function. Let \(n = d_L\) denote the number of classes,
\(d = d_0\) the number of input features, and
\(\tbf{a}_0 = \tbf{x} \in \mathbb{R}^{d}\). The per example crossentropy loss
\(\ell\) is given by
\(\ell(\hat{\tbf{y}}, \tbf{y}) = -\sum_{i=1}^n y_i \log \hat{y}_i\). Observe
in the \(L=1\) case, the network is equivalent to multi-class logistic
regression.

\begin{thm}[\(b=1\) Data Forging Fully Connected Neural Networks]
  \label{thm:b=1-data-forging}
  Let
  \(\mathcal{Z} := \mathbb{R}^d \times \{ u \in \{0,1\}^n : \sum_{i=1}^n u_i= 1\}\). For any two
  training examples
  \((\tbf{x}, \tbf{y}), (\tbf{x}', \tbf{y}') \in \mathcal{Z}\) if
  \(\nabla_{\theta}\ell(f_{\theta}(\tbf{x}), \tbf{y}) = \nabla_{\theta}\ell(f_{\theta}(\tbf{x}'),
  \tbf{y}')\), then \(\tbf{x} = \tbf{x}'\) and \(\tbf{y} = \tbf{y}'\).
\end{thm}

Proposition \ref{thm:b=1-data-forging} states that no two distinct
training examples with one hot labels can produce the same gradient,
an intuitive result and one that may point towards the difficulty of
exact data forging. Further, proposition \ref{thm:b=1-data-forging}
proves the stronger result of impossibility of exact data forging when
\(b=1\) for distinct inputs in the infinite domain
\(\tbf{x}, \tbf{x}' \in \mathbb{R}^d\), not just the finite set of valid
discrete-valued images (see appendix \ref{sec:proof-proposition-1} for
the full proof). However, models are typically trained with a batch
size \(b \gg 1\), therefore an investigation of the possibility of exact
data forging within this regime is required.

\subsection{Case of \(b > 1\)}

In contrast to the previous setting, we prove that there can exist
distinct mini-batches of size \(b>1\) that produce the same gradient.

\begin{thm}[\(b >1\) Data Forging Fully Connected Neural Networks]
  \label{thm:b1-data-forging}
  Given a fully connected neural network and mini-batch
  \( B = \{(\tbf{x}^{(k)},\tbf{y}^{(k)})\}_{k=1}^b\), if
  \(db > d_1(d + b)\), where \(d_1\) is the size of the first hidden
  layer, then there exists an infinite number of perturbation matrices
  \(\tbf{P} \in \mathbb{R}^{d\times b}\) such that for the distinct mini-batch
  \(B' = \{(\tbf{x}^{(k)} + [\tbf{P}]^k, \tbf{y}^{(k)})\}_{k=1}^b\),
  \(\nabla_{\theta}\mathcal{L}(B) = \nabla_{\theta}\mathcal{L}(B')\).
\end{thm}

Proposition \ref{thm:b1-data-forging} proves the existence of
perturbations that can be applied to the original training examples
\((\tbf{x},\tbf{y}) \in B\) such that the same gradient is still
preserved (see Fig. \ref{fig:forged-analytic} for an example).

\begin{figure}[htb]
  \centering
  \begin{subfigure}[b]{0.055\textwidth}
    \centering
    \includegraphics[width=\textwidth]{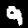}
  \end{subfigure}
  ~
  \begin{subfigure}[b]{0.055\textwidth}
    \centering
    \includegraphics[width=\textwidth]{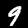}
  \end{subfigure}
  ~
  \begin{subfigure}[b]{0.055\textwidth}
    \centering
    \includegraphics[width=\textwidth]{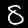}
  \end{subfigure}
  ~ \dots
  \begin{subfigure}[b]{0.055\textwidth}
    \centering
    \includegraphics[width=\textwidth]{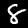}
  \end{subfigure}
  \vfill
  \begin{subfigure}[b]{0.055\textwidth}
    \centering
    \includegraphics[width=\textwidth]{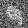}
  \end{subfigure}
  ~
  \begin{subfigure}[b]{0.055\textwidth}
    \centering
    \includegraphics[width=\textwidth]{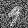}
  \end{subfigure}
  ~
  \begin{subfigure}[b]{0.055\textwidth}
    \centering
    \includegraphics[width=\textwidth]{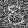}
  \end{subfigure}
  ~ \dots
  \begin{subfigure}[b]{0.055\textwidth}
    \centering
    \includegraphics[width=\textwidth]{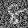}
  \end{subfigure}

  \caption{\tbf{Bottom:} The original mini-batch \(B\) consisting of
    $b = 32$ examples from MNIST. \tbf{Top:}. A forged mini-batch
    \(B'\) with an approximation error of \(\approx 10^{-7}\). In this
    setting, the model is an \(L=2\) ReLU neural network with 8 hidden
    units. Similar mini-batches were constructed by Pan et
    al. \cite{pan2022exploring} within the context of gradient
    inversion.}
  \label{fig:forged-analytic}
\end{figure}


Exact data forging can be performed using this \emph{mini-batch
  perturbation} approach.  However, while the gradients match, the
forged input examples \(\tbf{x}'^{(k)}\) are not guaranteed to be
members of the set of valid images
\(\{ \frac{q}{255}, q \in [0..255]\}^d\); the applied perturbation may
result in a mini-batch that is outside the allowed input domain. While
we prove the existence of an infinite number of possible perturbations
that may applied to the original mini-batch that preserve the
gradient, we have no guarantee that a perturbation that preserves
image validity exists, even for relatively simple linear
models. Efficiently finding such a perturbation is a non trivial task,
and is currently the main barrier to successful exact data forging.



\subsection{Exact Data Forging for \(L=1\) Networks}

Recall that for a \(L=1\) network with single parameter matrix
\(\tbf{W}\), the gradient of crossentropy loss is given by
\(\nabla_{\tbf{W}} \mathcal{L}(B) = \frac{1}{b}
\tbf{X}(f_{\tbf{W}}(\tbf{X})\mbox{diag}(\tbf{v}) - \tbf{Y})^T\), where
\(\tbf{v}_k = \sum_{j=1}^n \tbf{y}^{(k)}_j\), and that exact data forging
reduces to the problem of finding two matrices
\(\tbf{X}' \in \mathbb{R}^{d\times b}, \tbf{Y}' \in \mathbb{R}^{n\times b}\) where
\(\tbf{X}' \ne \tbf{X}\) such that
\(\tbf{X}(f_{\tbf{W}}(\tbf{X})\mbox{diag}(\tbf{v}) - \tbf{Y})^T
=\tbf{X}'(f_{\tbf{W}}(\tbf{X}')\mbox{diag}(\tbf{v}') - \tbf{Y}')^T\).

\begin{thm}[\(L=1\) Data Forging]
  \label{thm:data-forging-with-2}  
  If \(b > \mbox{rank}(\nabla_{\tbf{W}} \mathcal{L}(B))\), then there exists an
  infinite number of distinct mini-batches
  \(B' \subset \mathbb{R}^d \times \mathbb{R}^n\) of size \(b\) such that
  \(\nabla_{\tbf{W}} \mathcal{L}(B') = \nabla_{\tbf{W}} \mathcal{L}(B)\).
\end{thm}

Proposition \ref{thm:data-forging-with-2} (see appendix
\ref{sec:tbfpr-prop-refthm:d} for the proof) extends the mini-batch
perturbation approach for \(L=1\) networks, showing that there can
exist two distinct mini-batches with the same gradient that are not
necessarily a perturbation of the other, and do not share the same
labels. We prove this with the following steps:

Let matrix
\(\tbf{D} = (f_{\tbf{W}}(\tbf{X})\mbox{diag}(\tbf{v}) - \tbf{Y})^T\)
be the error matrix of the model on input \(\tbf{X}\). Our first
method of exact data forging exploits a fundamental property of this
matrix (for proof see Appendix \ref{sec:tbfpr-lemma-refl}).

\begin{lem}
  \label{lem:sumzero}
  Let
  \(\tbf{D} = (f_{\tbf{W}}(\tbf{X})\mbox{diag}(\tbf{v}) -
  \tbf{Y})^T\). For any two matrices
  \(\tbf{X} \in \mathbb{R}^{d\times b}\) and
  \(\tbf{Y} \in \mathbb{R}^{n\times b}\), we have that
  \(\sum_{j=1}^n \tbf{D}_{ij} = 0\), \(\forall i, \ 1 \le i \le b\).
\end{lem}

Using this property, we may devise the following \emph{error matrix
  sampling } approach towards exact data forging for \(L=1\)
networks. Given mini-batch \(B\) of size \(b\) and gradient
\(\nabla_{\tbf{W}}\mathcal{L}(B)\) of rank \(r\): \emph{(i)} sample any matrix
\(\tbf{D} \in \mathbb{R}^{b\times n}\) of such that
\(\sum_{j=1}^n \tbf{D}_{ij} = 0\) and \(\mbox{rank}(\tbf{D}) \ge r\),
\emph{(ii)} solve for \(\tbf{X}'\) the matrix equation
\(\tbf{X}'\tbf{D} = b\nabla_{\tbf{W}}\mathcal{L}(B)\), \emph{(iii)} set
\(\tbf{Y}' := f_{\tbf{W}}(\tbf{X}')\mbox{diag}(\tbf{v}) - \tbf{D}^T\),
and \emph{(iv)} construct forged mini-batch
\(B' = \{(\tbf{x}'^{(k)},\tbf{y}'^{(k)})\}_{k=1}^b\) where each
\(\tbf{x}'^{(k)}, \tbf{y}'^{(k)}\) are the \(k\)-th column of
\(\tbf{X}',\tbf{Y}'\) respectively.

Figure \ref{fig:forged-ems} gives two such examples of distinct
mini-batches that produce the same gradient for both the MNIST and
CIFAR10 datasets for an \(L=1\) network. The above error matrix
sampling approach is not restricted to constructing forged
mini-batches with the same labels as the original, proving the
existence of distinct mini-batches with distinct labels can produce
the same gradient. The infinite number of possible error matrices
\(\tbf{D}\) as well as the potentially infinite solutions to equation
\(\tbf{X}'\tbf{D} = b\nabla_{\tbf{W}} \mathcal{L}(B)\) capture all the distinct
mini-batches \(B' \subset \mathbb{R}^d \times \mathbb{R}^n\) such that
\(\nabla_{\tbf{W}} \mathcal{L}(B) = \nabla_{\tbf{W}} \mathcal{L}(B')\). However, analogous to the
mini-batch perturbation approach, choosing the correct error matrix
\(\tbf{D}\) such that the mini-batch \(B'\) that is constructed from
the resulting \(\tbf{X}'\) and \(\tbf{Y}'\) falls within the required
domain is a non trivial, computationally infeasible task.

Our analysis takes the first step towards understanding exact data
forging. Within the unrestricted mini-batch domain, we prove that
exact data forging is indeed possible and provide methods of
constructing forged mini-batches with approximation errors that are
indistinguishable from reproduction errors. Further research and
analysis are required in order to determine whether data forging
within these restricted domains is possible. Given the difficulty of
this task, we conjecture that it may not be computationally feasible.



\begin{figure}[htb]
  \centering
  \begin{subfigure}[b]{0.055\textwidth}
    \centering
    \includegraphics[width=\textwidth]{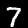}
  \end{subfigure}
  ~
  \begin{subfigure}[b]{0.055\textwidth}
    \centering
    \includegraphics[width=\textwidth]{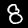}
  \end{subfigure}
  ~
  \begin{subfigure}[b]{0.055\textwidth}
    \centering
    \includegraphics[width=\textwidth]{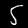}
  \end{subfigure}
  ~ \dots
  \begin{subfigure}[b]{0.055\textwidth}
    \centering
    \includegraphics[width=\textwidth]{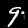}
  \end{subfigure}
  \vfill
  \begin{subfigure}[b]{0.055\textwidth}
    \centering
    \includegraphics[width=\textwidth]{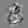}
  \end{subfigure}
  ~
  \begin{subfigure}[b]{0.055\textwidth}
    \centering
    \includegraphics[width=\textwidth]{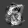}
  \end{subfigure}
  ~
  \begin{subfigure}[b]{0.055\textwidth}
    \centering
    \includegraphics[width=\textwidth]{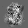}
  \end{subfigure}
  ~ \dots
  \begin{subfigure}[b]{0.055\textwidth}
    \centering
    \includegraphics[width=\textwidth]{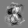}
  \end{subfigure}
  \vfill \mbox{  } \vfill
  \begin{subfigure}[b]{0.055\textwidth}
    \centering
    \includegraphics[width=\textwidth]{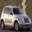}
  \end{subfigure}
  ~
  \begin{subfigure}[b]{0.055\textwidth}
    \centering
    \includegraphics[width=\textwidth]{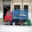}
  \end{subfigure}
  ~
  \begin{subfigure}[b]{0.055\textwidth}
    \centering
    \includegraphics[width=\textwidth]{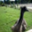}
  \end{subfigure}
  ~ \dots
  \begin{subfigure}[b]{0.055\textwidth}
    \centering
    \includegraphics[width=\textwidth]{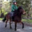}
  \end{subfigure}
  \vfill
  \begin{subfigure}[b]{0.055\textwidth}
    \centering
    \includegraphics[width=\textwidth]{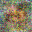}
  \end{subfigure}
  ~
  \begin{subfigure}[b]{0.055\textwidth}
    \centering
    \includegraphics[width=\textwidth]{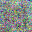}
  \end{subfigure}
  ~
  \begin{subfigure}[b]{0.055\textwidth}
    \centering
    \includegraphics[width=\textwidth]{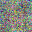}
  \end{subfigure}
  ~ \dots
  \begin{subfigure}[b]{0.055\textwidth}
    \centering
    \includegraphics[width=\textwidth]{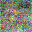}
  \end{subfigure}

  \caption{A mini-batch from MNIST (top) and CIFAR10 (bottom) and
    distinct forged counterpart produced by our error matrix sampling
    approach. Both examples produce an approximation error on the
    order of \(10^{-7}\), 5 orders of magnitude smaller than existing
    attacks.}
  \label{fig:forged-ems}
\end{figure}

\section{Related Work}
\subsection{Proof of Learning}
Jia et al. \cite{jia2021proof} observed that the production and
verification of an execution trace is tantamount to proving one has
trained model \(\theta_T\) on dataset \(D = \bigcup_{i=0}^{T-1} B_i\). They
formalise execution traces as \emph{Proof of Learning} (PoL)
sequences, and verify the final model \(\theta_T\) is indeed the result of
the optimisation captured in the trace by recomputing the checkpoints
using the provided mini-batches and previous checkpoint, ensuring that
each recomputed checkpoint is within some pre-defined distance
threshold \(\epsilon\), e.g. \(\ell_2\) norm, from what is reported in the
trace.

Implicit within the definition of PoL is the uniqueness of execution
traces for a given final model \(\theta_T\), and that the construction of a
distinct trace for \(\theta_T\) that passes verification is at least as
difficult a task as training. Consequently, Jia et al. consider PoL as
a viable means of proving a third party has indeed expended the energy
to train the model and therefore owns it. Constructing distinct traces
that challenge this assumption is known as PoL ``spoofing'', i.e.,
producing a trace that differs from the original by employing a
different sequence of mini-batches that follow a different training
trajectory, but ultimately end at the same final model \(\theta_T\). Given
such a spoofed execution trace \(S'\) for original trace \(S\), let
\(D\) and \(D'\) be the differing datasets used in \(S\) and \(S'\)
respectively. One immediate consequence of successful PoL spoofing
regards the training examples \((\tbf{x},\tbf{y}) \notin D \cap D'\); namely
the question arises of \emph{``which are the real training examples
  that trained \(\theta_T\)"}?  PoL spoofing results in the apparent
contradiction that a model was both trained and not trained on some
set of training examples, in effect \emph{deleting} them from a
model's training set. 

\subsection{PoL Adversarial Examples}
\label{sec:pol-advers-exampl}

Zhang et al. \cite{zhang2022adversarial} challenged the underlying
assumption of PoL by constructing ``adversarial examples'' to PoL
\emph{i.e.,} two distinct execution traces for the same final model
\(\theta_T\). Zhang et al.'s proposed adversarial attack operates by
initialising dummy checkpoints such that the distance between
consecutive checkpoints is less than the chosen verification distance
threshold \(\epsilon\), and optimise ``random'' noise to add to mini-batches
consisting of public data such that the gradient is close to zero,
ensuring the next model checkpoint lies within the \(\epsilon\)-ball of the
next dummy checkpoint, and thus passing verification. However, Fang et
al. \cite{fang2023proof} find that they are fragile to the
hyperparameters of the verification process. Namely, under small
checkpointing intervals, the attacks break down. Additionally, Zhang
et al. apply their attack to image classification models, and there is
no mention nor guarantee of whether the final synthesized images are
still within the problem domain \emph{e.g.,} whether the final pixel
values are within \(0-255\). Any training examples outside the problem
domain in the mini-batches of an execution trace immediately point to
a detectable spoof. Fang et al. provide an ``infinitesimal update''
attack against PoL that is conceptually similar to Zhang et al.'s,
where a learning rate is chosen to be sufficiently small such that
regardless of the training data, the next model checkpoint lies within
the \(\epsilon\)-ball of the dummy checkpoint. The success of these attacks
that exploit the static verification threshold and their consequences
on the robustness of PoL verification is discussed in
\cite{fang2023proof}.  One limitation of these attacks is that they
require foreknowledge of the chosen \(\epsilon\) by the verifier, something
that may not be realistic in practice, and the utilisation of
infinitesimally small learning rates may also be suspect to a
verifier.

Existing PoL adversarial attacks
\cite{zhang2022adversarial,fang2023proof} exploit near zero norm
gradients to produce distinct execution traces that can pass PoL
verification. However it remains an open question whether two distinct
execution traces with distinct checkpoints may be constructed for
model \(\theta_T\) with valid, honest gradient updates that do not require
near zero norm gradients. Fang et al. evaluate attempts at doing so
via data ordering attacks \cite{shumailov2021manipulating}
\emph{i.e.,} finding a distinct sequence of valid mini-batches that
take random initialisation \(\theta_0\) to the desired \(\theta_T\); however
these were also found to be unlikely to succeed.

\section{Conclusion}
Answers to the question of exact data forging within the restricted
problem domain will have wide reaching ramifications for the machine
learning and privacy community. Proving its feasibility and the
development of methods of performing such a task allows for stealthy
data forging; an attack which fully realises the serious data privacy
and security implications recognised by the prior work
\cite{kong2023can,thudi2022necessity,zhangverification}. Membership
inference attacks are the bread-and-butter of machine learning
privacy. They are used to perform ML auditing and compliance, two
tasks that are key in any real world deployment of an ML
model. Additionally, machine unlearning is currently the only adequate
response to GDPR's \emph{right to be forgotten}, allowing model owners
to comply with the law without needing to throw away the entire
model. Stealthy data forging has the potential to disrupt these two
tools, casting doubt on their methodology and veracity. These are
serious implications for the ML privacy community, however, they have
not been fully realised by current data forging attacks. Further
research is needed on whether stealthy data forging is indeed
possible; and whether more powerful attack algorithms exist. We have
begun by first conclusively determining that current attacks are not
stealthy and are detectable, cancelling any implications they may have
on ML privacy. Additionally, our theoretical analysis takes a first
step towards understanding data forging and analysing the conditions
where different mini-batches can produce the same update.

\section{Ethics Considerations}

The danger of data forging attacks is that they appear to be the
perfect tool to evade data auditing and compliance
regulations. Malicious model owners who have trained their models on
copyrighted, and or sensitive personal data can \emph{forge} these
datasets into legally compliant ones, and publicly claim that the
forgery is in fact the true dataset. As a result, these problematic
models can remain operational, furthering the risk of potential harm
to the data owners. This can come in the form of copyright
infringement, and or the unconsented disclosure of sensitive data to
3rd parties. Research that furthers the development of data forging
must contend with potential harm these attacks can facilitate.

However, in our work, we have shown that attempts at data forging are
easily detectable by an informed adversary, and that, in their current
instantiation, they are not a real threat against data
governance. From an ethical point of view, our work can be seen as a
furthering an optimistic view of the future risk of evading data
regulation via data forging. We have seen that current attacks are not
able to fully realise their serious implications on data governance,
and our theoretical results suggests that developing attacks that can
is a non trivial, computationally infeasible task.

In carrying out our research, we believe there are no ethical
issues. Our results are the result of training on public datasets
(MNIST and CIFAR10), and we replicate data forging attacks in a closed
environment on execution traces resulting from these public datasets.

\section{Open Science}

We are committed to the open science policy by providing all our code
at the following \href{https://zenodo.org/records/15594394}{link}.

\bibliography{references}
\bibliographystyle{plain}

\appendix
\input{appendix.tex}

\end{document}

%% file: math_commands.tex
\newcommand{\tbf}{\textbf}
\newcommand{\nabbi}{\nabla_{\theta} \mathcal{L}(B_i; \theta_i)}

\newcommand{\del}{\partial}

\newcommand{\diag}{\mbox{diag}}
\newcommand{\rank}{\mbox{rank}}
\newcommand{\bdelta}{\boldsymbol{\delta}}
\newcommand{\vvec}{\mbox{vec}}

\newtheorem{thm}{Proposition}
\newtheorem{lem}{Lemma}
\newtheorem{cor}{Corollary}

\newtheorem*{thm*}{Proposition}
\newtheorem{defn}{Definition}

%% file: appendix.tex
\section{\tbf{Proofs}}

\subsection{\tbf{Proof of Proposition \ref{thm:b=1-data-forging}}}
\label{sec:proof-proposition-1}
We first introduce the
background, namely the neural network setup, the definition of the
loss function, and derive expressions for the gradients of the loss
function with respect to the network parameters. Then, we introduce and
prove some lemmas that are needed for the full proof. Finally the full
proof is given.

\tbf{Background.} A \(L\)-layer fully connected neural network consists
of weights and biases
\(\theta = [\tbf{W}_1, \tbf{b}_1 \tbf{W}_2, \tbf{b}_2, ..., \tbf{W}_L,
\tbf{b}_L]\) where each
\(\tbf{W}_i \in \mathbb{R}^{d_{i-1} \times d_i}, \tbf{b} \in \mathbb{R}^{d_i}\). The output of
the neural network is given by
\( f_{\theta}(\tbf{x}) = \mbox{softmax}(\tbf{z}_L)\), where
\( \tbf{z}_L = \tbf{W}_L^T \tbf{a}_{L-1} + \tbf{b}_L\),
\(\tbf{a}_i = h(\tbf{z}_i)\),
\(\tbf{z}_i = \tbf{W}_i^T\tbf{a}_{i-1} + \tbf{b}_i\), and \(h\) is the
activation function. Let \(n = d_L\) denote the number of classes,
\(d = d_0\) the number of input features, and
\(\tbf{a}_0 = \tbf{x} \in \mathbb{R}^{d}\). The per example crossentropy loss
\(\ell\) is given by
\(\ell(\hat{\tbf{y}}, \tbf{y}) = -\sum_{i=1}^n y_i \log \hat{y}_i\). Observe
in the \(L=1\) case, the network is equivalent to multi-class logistic
regression. \(\nabla_{\tbf{W}_i} \ell\) is given by the outer product of the
vectors \(\tbf{a}_{i-1}\) and \(\bdelta_i\):

\begin{equation}
  \label{eq:4}
  \nabla_{\tbf{W}_i} \ell = \tbf{a}_{i-1} \bdelta_i^T
\end{equation}

where
\(\bdelta_i = \frac{\del \ell}{\del \tbf{z}_i} \in \mathbb{R}^{d_i}\). Additionally,
\(\nabla_{\tbf{b}_i} \ell = \bdelta_i\). We have that

\begin{equation}
  \label{eq:3}
  \bdelta_i =\diag(\tbf{h}_i)\tbf{W}_{i+1}\bdelta_{i+1},
\end{equation}

where
\(\tbf{h}_{i} = [h'(z^{({i})}_1), ..., h'(z^{({i})}_{d_{i}})]^T \in
\mathbb{R}^{d_{i}}\). Let \(z_j = [\tbf{z}_L]_j\),
\( y_i = [\tbf{y}]_i\) and \(\hat{y}_j = [f_\theta(\tbf{x})]_j\). We can write the
\(j\)-th element of \(\bdelta_L\) as the following:

\begin{equation}\label{eq:dldzj}
  \begin{aligned}
    [\bdelta_L]_j = \frac{\del \ell}{\del z_j} &= - \sum_{i=1}^n y_i \cdot \frac{\del \log \hat{y}_i}{\del z_j} = - \sum_{i=1}^n \frac{y_i}{\hat{y}_i} \cdot \frac{\del \hat{y}_i}{\del z_j} \\
  \end{aligned}
\end{equation}
We know that the derivative of the softmax function is given by

\begin{equation}
  \frac{\del \hat{y}_i}{\del z_j} = \begin{cases}
                          \hat{y}_j(1 - \hat{y}_j) & \mbox{if } i = j \\
                          -\hat{y}_i \cdot \hat{y}_j & \mbox{otherwise},
                        \end{cases}
\end{equation}

which allows us to rewrite Equation \ref{eq:dldzj} as
\begin{equation}
  \label{eq:7}
  \begin{aligned}
    [\bdelta_L]_j = \frac{\del \ell}{\del z_j} &= -y_j(1-\hat{y}_j) - \sum_{\substack{i=1 \\ i \neq j}}^n y_i\cdot(-\hat{y}_j) \\
    [\bdelta_L]_j &= -y_j  + \hat{y}_j\sum_{i=1}^n y_i  \\
  \end{aligned}
\end{equation}

From (\ref{eq:7}), we have that \(\bdelta_L = v\hat{\tbf{y}} - \tbf{y}\),
where \(v = \sum_{i=1}^n [\tbf{y}]_i\).

\vspace{5pt}
\noindent\tbf{Associated Lemmas and Corollaries}
\begin{lem}
  \label{lem:1}
  For any two training examples \((\tbf{x}, \tbf{y})\) and
  \((\tbf{x}', \tbf{y}')\), if
  \(\nabla_{\tbf{W}_1}\ell(\tbf{x}', \tbf{y}') = \nabla_{\tbf{W}_1}\ell(\tbf{x},
  \tbf{y})\) and
  \( \nabla_{\tbf{b}_1}\ell(\tbf{x}', \tbf{y}') = \nabla_{\tbf{b}_1}\ell(\tbf{x},
  \tbf{y})\), then \(\tbf{x} = \tbf{x}'\).
\end{lem}

\begin{proof}
  If
  \(\nabla_{\tbf{W}_1}\ell(\tbf{x}', \tbf{y}') = \nabla_{\tbf{W}_1}\ell(\tbf{x},
  \tbf{y})\), then \(\tbf{x}'\bdelta_1^{'T} = \tbf{x}\bdelta_1^T\).
  If
  \(\nabla_{\tbf{b}_1}\ell(\tbf{x}', \tbf{y}') = \nabla_{\tbf{b}_1}\ell(\tbf{x},
  \tbf{y})\), then \(\bdelta'_1 = \bdelta_1\), resulting in the
  equation \(\tbf{x}'\tbf{u}^T = \tbf{x}\tbf{u}^T \) where
  \(\tbf{u} = \bdelta'_1 = \bdelta_1\). Since we have an outer
  product, this is only satsified when \(\tbf{x} = \tbf{x}'\).

\end{proof}

\begin{lem}
  \label{lem:2}
  For any two training examples \((\tbf{x}, \tbf{y})\) and
  \((\tbf{x}', \tbf{y}')\), if
  \(\nabla_{\tbf{W}_L}\ell(\tbf{x}', \tbf{y}') = \nabla_{\tbf{W}_L}\ell(\tbf{x},
  \tbf{y})\) and
  \( \nabla_{\tbf{b}_L}\ell(\tbf{x}', \tbf{y}') = \nabla_{\tbf{b}_L}\ell(\tbf{x},
  \tbf{y})\) then \(\tbf{y} = (v - v')\tbf{u} + \tbf{y}'\), where
  \(v = \sum_j [\tbf{y}]_j\), \(v' = \sum_j [\tbf{y}']_j\), and
  \(\tbf{u} \in \mathbb{R}^n\).
\end{lem}
\begin{proof}

  If
  \(\nabla_{\tbf{b}_L}\ell(\tbf{x}', \tbf{y}') = \nabla_{\tbf{b}_L}\ell(\tbf{x},
  \tbf{y})\), then \(\bdelta'_L = \bdelta_L\). If,
  \(\nabla_{\tbf{W}_L}\ell(\tbf{x}', \tbf{y}') = \nabla_{\tbf{W}_L}\ell(\tbf{x},
  \tbf{y})\), then
  \(\tbf{a}'_{L-1}\bdelta_L^{'T} = \tbf{a}_{L-1}\bdelta_L^T\). Since
  \(\bdelta_L' = \bdelta_L\) we have that
  \(\tbf{a}'_{L-1} = \tbf{a}_{L-1}\) Therefore,
  \(\tbf{z}_L' = \tbf{W}_L^T\tbf{a}_{L-1} + \tbf{b}_L =
  \tbf{W}_L^T\tbf{a}'_{L-1} + \tbf{b}_L=\tbf{z}_L\). Therefore
  \(\hat{\tbf{y}}' = \hat{\tbf{y}}\). Let
  \(\tbf{u} = \hat{\tbf{y}}' = \hat{\tbf{y}}\). Then

  \begin{equation}
    \label{eq:vv}
    \begin{aligned}
      \bdelta'_L &= \bdelta_L \\
      v'\hat{\tbf{y}}' - \tbf{y}' &= v\hat{\tbf{y}} - \tbf{y} \\
      v'\tbf{u} - \tbf{y}' &= v\tbf{u} - \tbf{y} \\
      \tbf{y} &= (v - v')\tbf{u} + \tbf{y}'
    \end{aligned}
  \end{equation}
\end{proof}

\begin{cor}
  \label{cor:data-forging-lemmas}
  For any two training examples \((\tbf{x}, \tbf{y})\) and
  \((\tbf{x}', \tbf{y}')\), if
  \(\nabla_{\tbf{W}_L}\ell(\tbf{x}', \tbf{y}') = \nabla_{\tbf{W}_L}\ell(\tbf{x},\tbf{y})\)
  and
   \( \nabla_{\tbf{b}_L}\ell(\tbf{x}', \tbf{y}') = \nabla_{\tbf{b}_L}\ell(\tbf{x},\tbf{y})\)
  and
   \( \sum_j[\tbf{y}]_j = \sum_j[\tbf{y}']_j\)
  then \(\tbf{y} = \tbf{y}'\).
\end{cor}
\begin{proof}
  If the first two conditions hold, then, by Lemma \ref{lem:2},
  \(\tbf{y} = (v - v')\tbf{u} + \tbf{y}'\), where
  \(v = \sum_j [\tbf{y}]_j\), \(v' = \sum_j [\tbf{y}']_j\). If these two
  sums are equal, then \(v - v' = 0\) which, substituting into
  \eqref{eq:vv} gives \(\tbf{y} = \tbf{y}'\).
\end{proof}

\tbf{Full Proof.} We now prove Proposition \ref{thm:b=1-data-forging}.

\begin{proof}
  The proof follows from Lemma \ref{lem:1} and Lemma \ref{lem:2}. From
  Lemma \ref{lem:1}, if
  \( \nabla_{\tbf{W}_1}\ell(\tbf{x}', \tbf{y}') = \nabla_{\tbf{W}_1}\ell(\tbf{x},
  \tbf{y})\) and
  \(\nabla_{\tbf{b}_1}\ell(\tbf{x}', \tbf{y}') = \nabla_{\tbf{b}_1}\ell(\tbf{x},
  \tbf{y})\), then \(\tbf{x} = \tbf{x}'\).

  From Lemma \ref{lem:2}, if
  \( \nabla_{\tbf{W}_L}\ell(\tbf{x}', \tbf{y}') = \nabla_{\tbf{W}_L}\ell(\tbf{x},
  \tbf{y})\) and
  \(\nabla_{\tbf{b}_L}\ell(\tbf{x}', \tbf{y}') = \nabla_{\tbf{b}_L}\ell(\tbf{x},
  \tbf{y})\), then \(\tbf{y} = (v - v')\tbf{u} + \tbf{y}'\). If
  \(\tbf{y}\) and \(\tbf{y}'\) are one hot labels, then \( v =
  v' = 1\). Consequently, from Corollary \ref{cor:data-forging-lemmas},
  \(\tbf{y}= \tbf{y}'\).
  
\end{proof}

\subsection{\tbf{Proof of Lemma \ref{lem:sumzero}}}
\label{sec:tbfpr-lemma-refl}
\begin{proof}
  Let \(\tbf{y} = [y_1,y_2,\dots,y_n]^T\) represent the \(i\)-th
  column of \(\tbf{Y}\), and \(\tbf{s} = [s_1,s_2,\dots,s_n]^T\)
  represent the \(i\)-th column of \(f_{\tbf{W}}(\tbf{X})\). We then have that
  \begin{equation}
    \begin{aligned}
      \sum_{j=1}^n D_{ij} &=  -\sum_{j=1}^n y_j + \sum_{j=1}^ns_j \cdot \sum_{k=1}^n y_k\\
                       &= -\sum_{j=1}^n y_j + \sum_{k=1}^n y_k = 0.
    \end{aligned}
  \end{equation}
\end{proof}

\subsection{\tbf{Proof of Proposition \ref{thm:b1-data-forging}}}
\label{sec:tbfpr-prop-refthm:b1}
We present the proof as follows: first we derive expressions for the
gradient of the average loss function \(\nabla \mathcal{L}\) for fully connected
neural networks, as done similarly in Appendix
\ref{sec:proof-proposition-1} for \(\nabla \ell\). We then go on to give the
full proof.

\tbf{Background.} For the mini-batch setting, we can derive matrix
expressions for \(\nabla_{\tbf{W}_i} \mathcal{L}\). Writing everything in matrix
notation, we have the input matrix
\(\tbf{X} \in \mathbb{R}^{d_0 \times b}\), where the \(k\)-th column denotes
\(\tbf{x}^{(k)}\). We can also write the batched output of the model
\(\tbf{Z}_L \in \mathbb{R}^{d_L \times b}\) as the following matrix multiplication:

\begin{equation}
  \tbf{Z}_L = \tbf{W}_{L}^T \tbf{A}_{L-1} + \tbf{B}_L,
\end{equation}

where each \(\tbf{A}_i \in \mathbb{R}^{d_i \times b}\) is the batched activation after
the \(i\)-th layer i.e., the \(k\)-th column of \(\tbf{A}_i\)
represents the the activation of the network after the \(i\)-th layer
for the \(k\)-th example in the mini-batch \(\tbf{a}_{i}^{(k)} := [ a^{(k,i)}_1  a^{(k,i)}_2  \dots a^{(k,i)}_{d_i}]^T\):

\begin{equation}
    \tbf{A}_i = h(\tbf{Z}_i) = h(\tbf{W}_i^T\tbf{A}_{i-1} + \tbf{B}_i),
\end{equation}

where \(\tbf{A}_0 = \tbf{X}\), and
\(\tbf{B}_i \in \mathbb{R}^{d_i \times b}\), which is a matrix where each of the
columns are the same bias vector for the \(i\)-th layer repeated. Let
the scalar \( z^{(k,i)}_m \) represents the the \(m\)-th logit value
at the \(i\)-th layer of the \(k\)-th training example in the
mini-batch, and let the scalar \( a^{(k,i)}_m \) be the corresponding
activation value.

We can write the \((mn)\)-th element of
\(\nabla_{\tbf{W}_i} \mathcal{L}(B)\) as the following:

\begin{equation}\label{eq:5}
  \begin{aligned}
    &[ \nabla_{\tbf{W}_i} \mathcal{L}(B) ]_{mn} = \frac{1}{b} \sum_{(\tbf{x}, \tbf{y}) \in B} [ \nabla_{\tbf{W}_i} \ell(\tbf{x},\tbf{y}) ]_{mn} \\
    &= \frac{1}{b} \sum_{k=1}^b a^{(k,i-1)}_m \cdot \frac{\del \ell^{(k)}}{\del z^{(k,i)}_n} \\
  \end{aligned}
\end{equation}

To compute the entire gradient \(\nabla_{\tbf{W}_i} \mathcal{L}(B)\), not just a
single element, we can view it as the matrix multiplication
\( b \nabla_{\tbf{W}_i} \mathcal{L}(B) = \tbf{A}_{i-1} \tbf{D}_i,\) where
\(\tbf{D}_i \in \mathbb{R}^{b \times d_i}\) and the \(k\)-th row denotes
\(\bdelta^{(k)}_i\) i.e., the error at the \(i\)-th layer for the
\(k\)-th example in the mini-batch. Each \(\tbf{D}_i\) is given by

\begin{equation}
  \label{eq:6}
  \tbf{D}_i = (\tbf{D}_{i+1}\tbf{W}_{i+1}^T) \odot H_i.
\end{equation}

Equation \eqref{eq:6} is the matrix version of \eqref{eq:3}. The
\(k\)-th row of \(H_i \in \mathbb{R}^{b\times d_i}\) denotes
\(\tbf{h}^{(k)}_i\), and \(\odot\) gives the hadamard
product. \(\nabla_{\tbf{b}_i} \mathcal{L}\) is given by the average of all the rows
of \(\tbf{D}_i\) e.g.

\begin{equation*}
  [\nabla_{\tbf{b}_i} \mathcal{L}]_j = \frac{1}{b} \sum_{k=1}^b [\tbf{D}_i]_{kj}
\end{equation*}

\tbf{Full Proof.} Finally we can prove Proposition
\ref{thm:b1-data-forging}.

\begin{proof}
  We prove Proposition \ref{thm:b1-data-forging} by proving that there
  exists a non-zero \emph{perturbation} matrix
  \(\tbf{P} \in \mathbb{R}^{d_0 \times b}\) such that for a given mini-batch
  \(B = \{ (\tbf{x}^{(k)}, \tbf{y}^{(k)})\}_{k=1}^b\), a corresponding
  forged mini-batch \(B'\) can be constructed from \(B\) and
  \(\tbf{P}\), where
  \(B' = \{ (\tbf{x}^{(k)} + [\tbf{P}]^k, \tbf{y}^{(k)})\}_{k=1}^b\)
  and
  \(\nabla_{\tbf{W}_i} \mathcal{L}(B) = \nabla_{\tbf{W}_i} \mathcal{L}(B')\) and
  \(\nabla_{\tbf{b}_i} \mathcal{L}(B) = \nabla_{\tbf{b}_i} \mathcal{L}(B')\) holds
  \(\forall i \in \{1,2,...,L\}\). The operator \([\cdot]^k\) returns the
  \(k\)-th column of the input matrix. The perturbation matrix
  \(\tbf{P}\) itself must satisfy the following 2 matrix equations:

  \begin{equation}
    \label{eq:2}
    \begin{cases}
      \tbf{P}\tbf{D}_1 = \tbf{0} \\
      \tbf{W}_1^T\tbf{P} = \tbf{0}
    \end{cases}
  \end{equation}

  Let \(\tbf{X}' := \tbf{X} + \tbf{P}\), where
  \(\tbf{X}' \in \mathbb{R}^{d_0 \times b}\) represents the input matrix of the forged
  mini-batch \(B'\) and \(\tbf{X}\) is the input matrix of mini-batch
  \(B\), where the \(k\)-th column of \(\tbf{X}\) and \(\tbf{X}'\)
  denote \(\tbf{x}^{(k)}\) and
  \(\tbf{x}'^{(k)} := \tbf{x}^{(k)} + [\tbf{P}]^k\)
  respectively. Considering the first equation
  \(\tbf{P}\tbf{D}_1 = \tbf{0}\), if \(\tbf{P}\) satisfies this
  equation, then
  \(b\nabla_{\tbf{W}_1} \mathcal{L}(B') = \tbf{X}'\tbf{D}_1 = \tbf{X}\tbf{D}_1 +
  \tbf{P}\tbf{D}_1 = \tbf{X}\tbf{D}_1 = b\nabla_{\tbf{W}_1} \mathcal{L}(B).\)

  Thus satisfying the first equation ensures that the gradient for the
  first weight matrix \(\tbf{W}_1\) will match.

  Let \(\tbf{Z}'_1\) be the batched raw logit output of the first
  layer, i.e. \(\tbf{Z}'_1 := \tbf{W}_1^T\tbf{X}' + \tbf{B}_1\).  If
  \(\tbf{P}\) satisfies the second equation, then
  \( \tbf{Z}'_1 =\tbf{W}_1^T\tbf{X} + \tbf{W}_1^T\tbf{P} + \tbf{B}_1 =
  \tbf{W}_1^T\tbf{X} + \tbf{B}_1 = \tbf{Z}_1 \). Therefore,
  \(\tbf{A}'_1 = \tbf{A}_1\). Since \(\tbf{A}'_1 = \tbf{A}_1\), then
  \(\tbf{A}'_i = \tbf{A}_i \ \forall i \in \{1,2,\dots,L\}\). Satisfying the
  second equation ensures the activations for the first layer as well
  as all subsequent layers will be the same as that of the original
  mini-batch. If we use the same label i.e., each
  \(\tbf{y}'^{(k)} = \tbf{y}^{(k)}\), then, for the \(k\)-th training
  example in the forged mini-batch, we have that by
  \(\hat{\tbf{y}}'^{(k)} = \hat{\tbf{y}}^{(k)}\), and consequently
  \(\bdelta'^{(k)}_L = \bdelta^{(k)}_L\). Therefore
  \(\tbf{D}'_L = \tbf{D}_L\), and by \eqref{eq:6} we can show
  inductively that
  \(\tbf{D}'_i = \tbf{D}_i, \ \forall i \in \{1,2,...,L\}\). We have shown the
  base case \(i=L\). To show for \( i= L-1\), we first recognise that
  if \(\forall i\), \(A'_i = A_i\), then \(H'_i = H_i, \ \forall i\). By
  \eqref{eq:6} we then have
  \(\tbf{D}'_i = (\tbf{D}'_{i+1}\tbf{W}_{i+1}^T) \odot H'_i =
  (\tbf{D}_{i+1}\tbf{W}_{i+1}^T) \odot H_i = \tbf{D}_i \).  As a result,
  \(\nabla_{\tbf{b}_i} \mathcal{L}(B') = \nabla_{\tbf{b}_i} \mathcal{L}(B), \ \forall i \in
  \{1,2,\dots,L\}\). Regarding the gradients for the weight matrices
  \(\tbf{W}_i\) we then have that
  \( b\nabla_{\tbf{W}_i} \mathcal{L}(B')= \tbf{A}'_{i-1}\tbf{D}'_i =
  \tbf{A}_{i-1}\tbf{D}_i = b\nabla_{\tbf{W}_i} \mathcal{L}(B)\).  As a result
  \(\nabla_{\tbf{W}_i} \mathcal{L}(B') = \nabla_{\tbf{W}_i} \mathcal{L}(B), \ \forall i \in
  \{1,2,\dots,L\}\). Therefore, if the perturbation matrix \(\tbf{P}\)
  satifies the 2 matrix equations given in \eqref{eq:2}, then \(B\)
  and \(B'\) will produce the same gradient.

  In order to construct \(\tbf{P}\), we first observe that we can
  write the equations in \eqref{eq:2} as
  \( \tbf{I}\tbf{P}\tbf{D}_1 = \tbf{0}\) and
  \(\tbf{W}_1^T\tbf{P}\tbf{I} = \tbf{0}\), where \(\tbf{I}\) is the
  identity matrix of appropriate shape. We also have, by the ``vec
  trick'', that for any matrices \(A, X, B\) and \(C\),
  \(AXB = C \iff (B^T \otimes A)\vvec(X) = \vvec(C)\), where
  \(\otimes\) is the Kronecker product, and \(\vvec(\cdot)\) returns the vector
  after stacking all the columns of the input matrix vertically. We
  can rewrite our two equations as the following two linear systems:

  \begin{equation*}
    \begin{cases}
      (\tbf{D}_1^T \otimes \tbf{I})\vvec(\tbf{P}) = \vvec(\tbf{0}) \\
      (\tbf{I}^T \otimes \tbf{W}_1^T)\vvec(\tbf{P}) = \vvec(\tbf{0})
    \end{cases}
  \end{equation*}

  The first coefficient matrix
  \((\tbf{D}_1^T \otimes \tbf{I}) \in \mathbb{R}^{d_1d_0 \times d_0b}\), and the second
  coefficient matrix
  \((\tbf{I}^T \otimes \tbf{W}_1^T) \in \mathbb{R}^{d_1b \times d_0b}\), so we can combine
  them into the single linear system by stacking them vertically:

  \begin{equation}
    \begin{bmatrix}
      \tbf{D}_1^T \otimes \tbf{I}\\
      \tbf{I}^T \otimes \tbf{W}_1^T
    \end{bmatrix} \vvec(\tbf{P}) = \vvec(\tbf{0})
  \end{equation}

  Let \(\tbf{K} \in \mathbb{R}^{(d_1d_0 + d_1b) \times d_0b}\) be the above coefficient
  matrix, then the general solution is given by:

  \begin{equation}
    \vvec(\tbf{P}) = (\tbf{I} - \tbf{K}^+\tbf{K})\tbf{F},
  \end{equation}

  where \(\tbf{F} \in \mathbb{R}^{d_0b}\) is an arbitrary matrix, and
  \(\tbf{K}^+\) is the psuedo-inverse of \(\tbf{K}\).

  For the homogenous linear system
  \(\tbf{K}\vvec(\tbf{P}) = \vvec(\tbf{0})\), there are
  \(d_0b - \rank(\tbf{K})\) linearly independent solutions. If
  \(\tbf{K}\) is full rank and square, there are no non-trivial
  solutions. In order to ensure at least 1 set of linearly dependent
  solutions (e.g. one free variable) \(d_0b > d_1d_0 + d_1b\) must
  hold. This ensures that \(d_0b - \rank(\tbf{K}) \ge 1\).
\end{proof}

\subsection{\tbf{Proof of Proposition \ref{thm:data-forging-with-2}}}
\label{sec:tbfpr-prop-refthm:d}


  


\begin{proof}
  Let \(\tbf{G} = \nabla_{\tbf{W}} \mathcal{L}(B;\tbf{W})\). Finding a batch
  \(B' = \{(\tbf{x}'^{(1)}, \tbf{y}'^{(1)}), (\tbf{x}'^{(2)},
  \tbf{y}'^{(2)}), ..., (\tbf{x}'^{(b)}, \tbf{y}'^{(b)})\}\) with
  \(b \ge \mbox{rank}(\tbf{G})\), such that
  \(\nabla_{\tbf{W}} \mathcal{L}(B';\tbf{W}) = \tbf{G}\) requires that for every
  \((i,j) \in [d] \times [n]\),

  \begin{equation}
    \frac{\del \ell'^{(1)}}{\del z'^{(1)}_j}x'^{(1)}_i + \frac{\del \ell'^{(2)}}{\del z'^{(2)}_j}x'^{(2)}_i + ... + \frac{\del \ell'^{(b)}}{\del z'^{(b)}_j}x'^{(b)}_i = b\tbf{G}_{ij},
  \end{equation}

  where each \(\ell'^{(k)} = \ell(\tbf{x}'^{(k)},\tbf{y}'^{(k)})\), and
  \(\tbf{z}'^{(k)} = \tbf{W}^T\tbf{x}'^{(k)}\).  We can restate the
  problem as finding matrices \(\tbf{X}' \in \mathbb{R}^{d \times b}\) and
  \(\tbf{D} \in \mathbb{R}^{b \times n}\) such that

  \begin{equation}
    \tbf{X}'\tbf{D} = b\tbf{G},
  \end{equation}

  The \(k\)-th column of \(\tbf{X}'\) represents \(\tbf{x}'^{(k)}\),
  and the \(k\)-th row of \(\tbf{D}\) represents
  \(\frac{\del \ell'^{(k)}}{\del \tbf{z}'^{(k)}}\). From Lemma \ref{lem:sumzero},
  the elements of every row of \(\tbf{D}\) must sum to zero.
  Construct a matrix \(\tbf{D}\) such that
  \(\mbox{rank}(\tbf{D}) = \mbox{rank}(\tbf{G})\), and the elements of
  every row of \(\tbf{D}\) sum to zero. Finding the corresponding
  \(\tbf{X}'\) amounts to solving the linear system

  \begin{equation}
    \tbf{D}^T\tbf{x}'_i = b\tbf{g}_i
  \end{equation}

  where \(\tbf{x}'_i\) and \(\tbf{g}_i\) are the transposed \(i\)-th
  row of \(\tbf{X}'\) and \(\tbf{G}\) respectively. For each \(i\), we
  know that
  \(\mbox{rank}(\tbf{D}^T) = \mbox{rank}(\tbf{D}^T | b\tbf{g}_i)\),
  therefore we can be certain that at least one solution exists.

  Finally, the batch of examples \(B'\) whose gradient
  \(\nabla_{\tbf{W}} \mathcal{L}(B'; \tbf{W}) = \tbf{G}\) can constructed from the matrices
  \(\tbf{X}'\) and \(\tbf{D}\). For every \(k \in [b]\),
  \(\tbf{x}'^{(k)}\) is the given by the \(k\)-th column of
  \(\tbf{X}'\), and as shown in Lemma \ref{lem:sumzero}, each
  \(\tbf{y}'^{(k)} = v^{(k)}\hat{\tbf{y}}^{(k)} - [\tbf{D}]^T_k\), for any
  constant \(v^{(k)} \in \mathbb{R}\), and \([\tbf{D}]_k\) returns the
  \(k\)-th row of \(\tbf{D}\).
\end{proof}

\begin{table*}[t]
  \centering
  \begin{tabular}{||c||c|c|c|c|c||}
    \hline
    &\(b=1\)&\(b=2\)&\(b=3\)&\(b=4\)&\(b=5\) \\
    \hline
    \(v = 2\)&\xmark (48)&\xmark (1128)&\xmark (17,296)&\xmark (194,580)&\xmark (1,712,304)\\
    \(v = 3\)&\xmark (243)& \xmark (29,403)& \xmark(2,362,041)& \qmark(141,722,460)&\qmark(\(6.77\times10^{9}\)) \\
    \(v=4\)&\xmark(768)&\xmark(294,528)&\qmark(75,202,816)&\qmark(\(1.43\times10^{10}\))&\qmark(\(2.2\times10^{12}\))\\
    \(v=5\)&\xmark(1,875)&\xmark(1,756,875)&\qmark(\(1.1\times10^9\))&\qmark(\(5.13\times10^{11}\))&\qmark(\(1.92\times10^{14}\))\\
    \(v=6\)&\xmark(3,888)&\xmark(7,556,328)&\qmark(\(9.79\times10^9\))&\qmark(\(9.51\times10^{12}\))&\qmark(\(7.38\times10^{15}\))\\
    \hline
  \end{tabular}
  \caption{For increasingly large values of \(v\) and \(b\) (\(d\) and
    \(n\) are fixed to \(d=4\) and \(n=3\)), we report, after
    searching through all the possible mini-batches, whose total
    number is given in brackets, whether a forged mini-batch that
    produced the same gradient was found. In every case, we did not
    find one after an exhaustive search. We give a \qmark \ in those
    setting where we did not conduct a search, due to the large number
    of possible mini-batches.}
  \label{tab:bf-forging}
\end{table*}

\section{\tbf{Accuracy Measurements}}
\label{sec:tbfacc-meas}
\begin{table*}
  \centering
  \begin{tabular}{| l | c | c| c | c | c ||}
    \hline
    Model/Dataset& \(b = 50\)& \(b=100\)&\(b=500\)&\(b=1000\)&\(b=2000\)\\
    \hline
    LeNet/MNIST   & \(0.94 \pm 2.4\mathrm{e}-4\) & \(0.94 \pm 1.6\mathrm{e}-4\) & \(0.95 \pm 4.7\mathrm{e}-5\) & \(0.95 \pm 0\) & \(0.95 \pm 8.1\mathrm{e}-5\) \\
    VGGmini/CIFAR10     & \(0.63 \pm 3.3\mathrm{e}-4\) & \(0.62 \pm 2.4\mathrm{e}-4\) & \(0.64 \pm 1.4\mathrm{e}-4\) & \(0.64 \pm 1.9\mathrm{e}-4\) & \(0.55 \pm 4.5\mathrm{e}-4\) \\
    FCN/Adult     & \(0.75 \pm 0\) & \(0.75 \pm 0\) & \(0.75 \pm 0\) & \(0.75 \pm 0\) & \(0.75 \pm 0\)\\
    Transformer/IMDB    & \(0.87 \pm 1.3\mathrm{e}-3\) & \(0.85 \pm 1.9\mathrm{e}-3\) & \(0.83 \pm 9.5\mathrm{e}-3\) & \(0.84 \pm 2.4\mathrm{e}-03\) & \(0.84 \pm 1.9\mathrm{e}-3\)\\
    \hline
  \end{tabular}
  \caption{Observed variance in accuracy from GPU auto-tuning non determinism}
  \label{tab:acc}
\end{table*}

\section{\tbf{Brute Force Forging}}
\label{sec:tbfbr-force-forg}

Real images consist of \(256\) possible pixel values. Given an image
consisting of just \(2\) pixels (with a single channel), \(3\)
possible classes, and a batch size of \(2\), this would result in
\({256^2 \times 3 \choose 2} \approx 1.9\times10^{10}\) mini-batches. Computing the
gradient of and comparing that many mini-batches is not practical; so
an evaluation of the brute force approach is not possible when we
allow 256 possible values. However, when we consider much smaller
values for \(v\), the number of allowed values, namely just \(2\) and
\(3\), the task becomes much more feasible. In other words, when
\(v=2\), pixels may have only the values \(0\) or \(255\) i.e., on or
off. For \(v=3\), pixels may take the values, \(0, 127,\) or \(255\)
i.e., 3 possible values. This drastically reduces the number of
mini-batches that we need to search through. 

For a given setup i.e., values of \(d, b, v,\) and \(n\), we run the
following experiment:

\begin{enumerate}
\item Sample a true minibatch \(B\) from the possible
  \({v^d \times n \choose b}\) number of mini-batches.
\item Calculate it's gradient \(\nabla_{\theta} \mathcal{L}(B)\).
\item Search through all \({v^d \times n \choose b} - 1\) other
  mini-batches for a distinct mini-batch that produces the same
  gradient.
\end{enumerate}

We run this experiment on a logistic regression model, as well as a
\(L=2\) fully connected neural network with \(10\) hidden units and
ReLU activation function We additionally fix \(d=4\) and
\(n=3\). Table \ref{tab:bf-forging} summarises our results.

\section{\tbf{Comparison of Data Forging Attacks}}
\label{sec:tbfc-data-forg} 
In Figure \ref{fig:forging-frac-thudi-zhang}, we plot the performance
of both Thudi et al. \cite{thudi2022necessity} and Zhang et
al.~\cite{zhangverification}'s data forging attacks, for both
LeNet/MNIST and VGGmini/CIFAR10 execution traces.
\begin{figure}[!htb]
  \centering
  \begin{subfigure}[t]{0.23\textwidth}
    \centering
    \includegraphics[width=\textwidth]{ ./images/final_v2/lenet-mnist-forging-frac_thudi.pdf }
    \caption{Thudi on LeNet/MNIST}
  \end{subfigure}
  ~
  \begin{subfigure}[t]{0.23\textwidth}
    \centering
    \includegraphics[width=\textwidth]{ 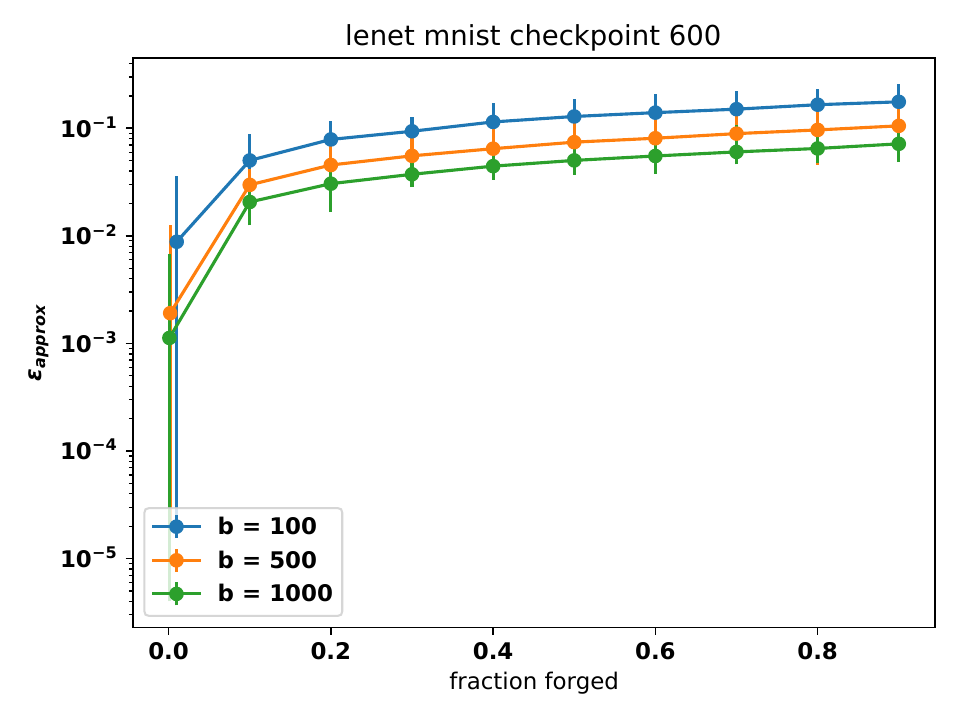 }
    \caption{Zhang on LeNet/MNIST}
  \end{subfigure}

  \begin{subfigure}[t]{0.23\textwidth}
    \centering
    \includegraphics[width=\textwidth]{ ./images/final_v2/vgg-cifar-forging-frac_thudi.pdf }
    \caption{Thudi onVGGmini/CIFAR10}
  \end{subfigure}
  ~
  \begin{subfigure}[t]{0.23\textwidth}
    \centering
    \includegraphics[width=\textwidth]{ 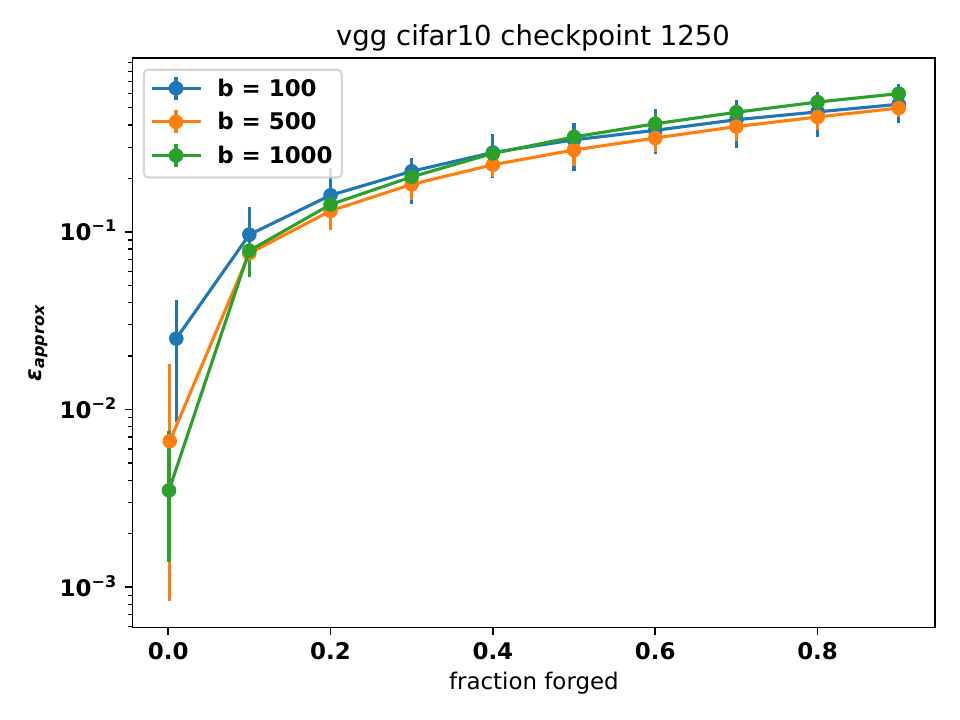 }
    \caption{Zhang on VGGmini/CIFAR10}
  \end{subfigure}
  \caption{We plot the average approximation error of both Thudi et
    al.'s~\cite{thudi2022necessity} and Zhang et al.'s~\cite{zhangverification} attack (over several runs), as well as
    the min and max observed approximation error for different batch
    sizes and increasing forging fractions. We see both produce
    similar performance.}
  \label{fig:forging-frac-thudi-zhang}
\end{figure}
